\documentclass[%
aps, prb,%
 amsmath,amssymb,
reprint,%
]{revtex4-1}

\usepackage{graphicx}
\usepackage{dcolumn}
\usepackage{bm}
\usepackage{amsthm}
\usepackage{hyperref}
\newtheorem{theorem}{Theorem}[section]
\newtheorem{corollary}{Corollary}[theorem]
\newtheorem{lemma}[theorem]{Lemma}

\theoremstyle{remark}

\theoremstyle{definition}
\newtheorem{definition}{Postulate}[section]
\begin{document}


\title{Area law in the exact solution of many-body localized systems}

\author{Evgeny Mozgunov}
 \email{mvjenia@gmail.com}
\affiliation{ Institute for Quantum Information and Matter,
  Department of Physics,
  California Institute of Technology
}%


\date{\today}

\begin{abstract}
Many-body localization was proven under realistic assumptions by constructing a quasi-local unitary rotation that diagonalizes the Hamiltonian (Imbrie, 2016). A natural generalization is to consider all unitaries that have a similar structure. We bound entanglement for states generated by such unitaries, thus providing an independent proof of area law in eigenstates of many-body localized systems. An error of approximating the unitary by a finite-depth local circuit is obtained. We connect the defined family of unitaries to other results about many-body localization (Kim et al, 2014), in particular Lieb-Robinson bound. Finally we argue that any Hamiltonian can be diagonalized by such a unitary, given it has a slow enough logarithmic lightcone in its Lieb-Robinson bound.

%
\end{abstract}

\maketitle

\section{\label{sec:level1}Introduction}

Physical models change their behavior in the presence of disorder ~--- a randomness in the terms of the corresponding Hamiltonian that's breaking the translation symmetry. First step in understanding the physics of disorder was made by Anderson\cite{mrAnderson}. He mathematically demonstrated a phenomenon of localization for non-interacting particles in a disordered potential. A quantum particle doesn't leave the vicinity of its starting point. This behavior of non-interacting particles can be thought of as a zeroth order approximation in the description of the interacting system. The question of convergence of any expansion around the non-interacting solution turned out to be a hard nut to crack. It was noted\cite{Tau} that an interacting system can be thought of as a single particle traveling in the Fock space, but the corresponding potential was not localizing in the Anderson sense. First positive answer was presented by Basko et al\cite{BAA}, which is when this problem became widely known as many-body localization (MBL). The achievement of this pioneering work was to find a realistic system (electrons in grains with a weak tunneling between them) and a range of parameters (weak interaction) where the diagrammatic expansion becomes tractable. The authors demonstrated zero conductivity at a finite temperature ~--- the system behaved as an Anderson insulator would.

An alternative route of investigation concentrated around disordered spin systems and conjectures by Huse et al\cite{huse}.
Spins-$1/2$ are described by Pauli matrices $\sigma_i^\alpha$. These degrees of freedon are normally fixed, and excitations do not always admit a quasiparticle description; what localization meant was open to debate. Huse et al conjectured that dressed ("tailed" here and below) spin operators $\tau^\alpha_i$ can be found that are conserved quantities of the dynamics. The set of those quasi-local conserved quantities $\tau^\alpha_i$ was conjectured to form an alternative (logical) operator basis. The coefficients of the mapping between the two thus is given by the trace norm tr$(\sigma_i^\alpha \tau^\beta_j)$ which decays exponentially according to quasi-locality of $\tau_i^\alpha$. This mapping allows to rewrite the spin Hamiltonian as a classical Hamiltonian (only containing $\tau_i^z$) in the logical basis. This logical Hamiltonian contains a non-local exponentially decaying interaction inherited from quasi-locality of $\tau_i^\alpha$.

All of the above conjectures turn out to be true for a specific spin system solved by Imbrie\cite{Imbrie}. His construction directly obtained the logical Hamiltonian as expected by Huse et al\cite{huse} by interleaving steps of perturbation theory with spatially sparse non-perturbative rotations. The exponential tails are adjusted as no decay happens in the regions containing those non-perturbative rotations, called resonances. But resonances are proven to be rare thus exponential decay holds on average. We will present the construction in more detail below, as it inspired our own definition of many-body localization. The path to showing quasi-locality of logical operators lies through a property called telescopic sum that we show. To our knowledge Imbrie's work\cite{Imbrie} is the only first principles proof for spin systems that serves as a beacon guiding other mathematical work\cite{annal} in the field.

Quite a lot of progress has been made by taking a different starting point than first principles (the Hamiltonian). Most notably, a Lieb-Robinson bound has been proven\cite{Isaac} from conjectures alike ones by Huse et al\cite{huse}. Lieb-Robinson bound characterizes how an operator evolves under the dynamics generated by the Hamiltonian. Normally the region on which the majority of the operator weight can be found (in terms of trace coefficients with the operator basis as shown above) grows linarly with evolution time. In MBL systems the operator spreads slower, only logarithmic in time. We show that is also true for the first principles solution.

The intuition about MBL was often informed by numerical simulations, in particular Nayak et al\cite{nayak} conjectured that a finite-depth local circuit can map an eigenstate approxiamtely to a product state. We will confirm that conjecture by providing a bound on the error of such mapping, and emphasize that a single unitary works equally well for all eigenstates. A circuit on top of a product state can be thought of as a matrix-product state. Some of the authors focused on a so called strong MBL, where they were able to prove\cite{frieza} a matrix-product state representation of every eigenstate. They also conjectured that a stronger Lieb-Robinson bound where the lightcone does not spread at all holds for the states at low energies. We find no evidence for that claim. 


\section{\label{sec:level2}Properties of the family of Imbrie's solutions}
Every Hamiltonian $H$ is characterized by its eigenvalues $H_{diag}$ and eigenvectors $U$:
\begin{equation}
UHU^{\dag}=H_{diag}
\end{equation}
 Imbrie\cite{Imbrie} presented a construction of $U = \prod_k U_k$ that works for a specific Hamiltonian on a 1d chain:
 \begin{equation}
 H=\gamma\sum_i\sigma^x_i +\sum_{i}h_i \sigma^z_i + J_i \sigma^z_i \sigma^z_{i+1}
 \end{equation}
 with $-1 \leq h_i,J_i\leq 1$ uniform i.i.d. random numbers and $\gamma\ll 1$. The structure of  $U = \prod_k U_k$ is illustrated in Figure \ref{fig:cir1}. We shall present the construction abstractly, without specifying all of the internal structure of $U_k$. We believe that for any MBL system a generalization of Imbrie's construction can be defined, so our presentation can be read both as a restatement of Imbrie's proof, and as a postulate about any system sufficiently deep in the MBL phase.

\begin{figure}
    \centering
    \includegraphics[width=0.8\columnwidth]{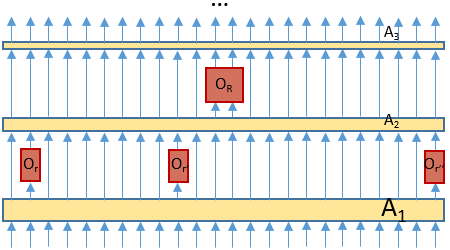}
    \caption{Definition of the circuit $U$}
    \label{fig:cir1}
\end{figure}
First consider no resonances (no red blocks in Figure \ref{fig:cir1}). There are infinitely many terms in the product  $U = \prod_k U_k$ counted by $k=1,2..$. 
\begin{definition}{\emph{(Support)}}
Each block
$
U_k =e^{A^{(k)}},
$
 where $iA^{(k)}$ is a Hermitean generator
 that is a sum of floor$(\frac{15}{7}L_k)$-local terms on each site of the chain, where $L_k =(15/8)^k$:
\begin{equation}
A^{(k)} = \sum_i A^{(k)}_i
\end{equation}
\end{definition}
Here locality is in terms of the interval on the chain that operator covers ~--- the distance in space between the leftmost and the rightmost spin that the operator depends on. We provide a guide to the text of Imbrie's paper in the Appendix B, explaining how we chose the number $15/7$ in the support definition above. We also conclude that Imbrie's proof allows one to restrict the norm of $\|A^{(k)}_i\|$ as follows:
\begin{definition}\label{Conj} {\emph{(Norm)}}
\begin{equation}
 \|A^{(k)}_i\| \leq c'\chi^{\textrm{ceil}L_{k-1}}
\end{equation}
where $\chi = c \gamma^{0.95}$ and $c, c'$ are constant numbers.
\end{definition}
Note that the magnetic field $\gamma$ enters the definition here. The construction rests on the fact that the $\frac{15}{7}L_{k}$ appearing in the support has the same scaling with $k$ as $L_{k-1} = \frac{8}{15}L_k$. The coefficient $c$ is obtained as a polynomial of $\textrm{exp(supp}A^{(k)}_i/L_k)$. If scalings were to be different even by a logarithmic factor (say, supports grow as $kL_k$), then instead of $(c\gamma^{0.95})^{L_{k-1}}$ the bound would contain $L_{k-1}! \gamma^{L_{k-1}}$, which diverges regardless how small $\gamma$ is. In Imbrie's paper the fact that support has $\sim L_k$ scaling is used several times, however the proof of it raises some questions, as we note in Appendix B. These two postulates are enough to prove the convergence of the no-resonance infinite product  $U = \prod_k U_k$ in the following sense:
\begin{theorem}[collar approximation]
Operators acquire exponentially decaying tails when conjugated by no-resonance $U$.
\begin{align}
\|UXU^{\dag}-U_cXU_c^{\dag}\| \leq 24\cdot 2^{0.5S(X)}\alpha^{c+1}\|X\|,\\ \|U^{\dag}XU-U_c^{\dag}XU_c\| \leq 24\cdot 2^{0.5S(X)}\alpha^{c+1}\|X\|
\end{align}
where $U_c$ is constructed using the generators $A^{(k)}$ truncated to a collar of $c$ sites of the support of $X$ (all the terms reaching outside the collar are dropped). $\alpha = 24c'\chi^ {\frac{56}{225}} $ is a power of transverse field $\gamma$.
\end{theorem}
For the proof, see Appendix A.

    \includegraphics[width=0.4\columnwidth]{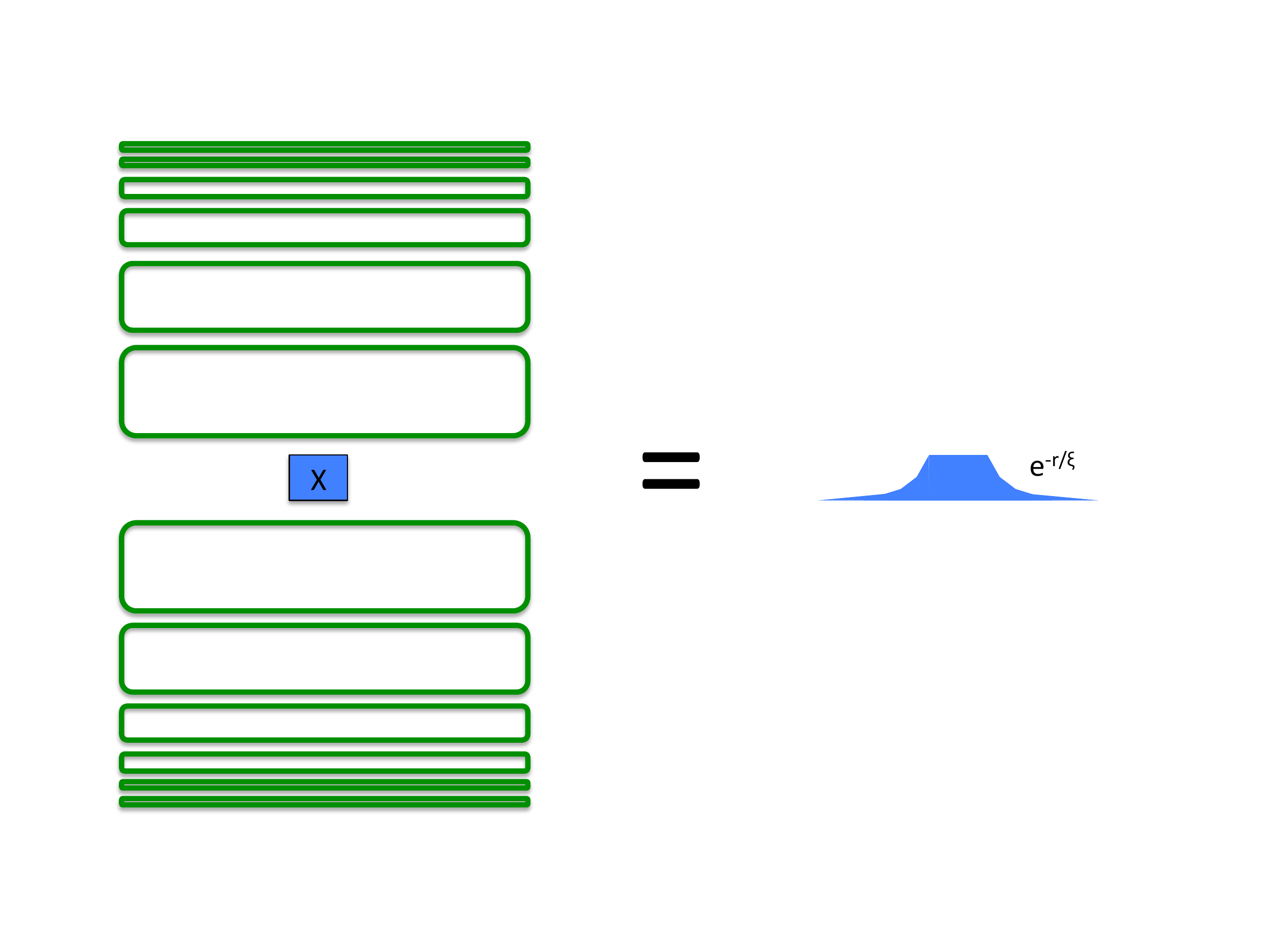}
    \includegraphics[width=0.4\columnwidth]{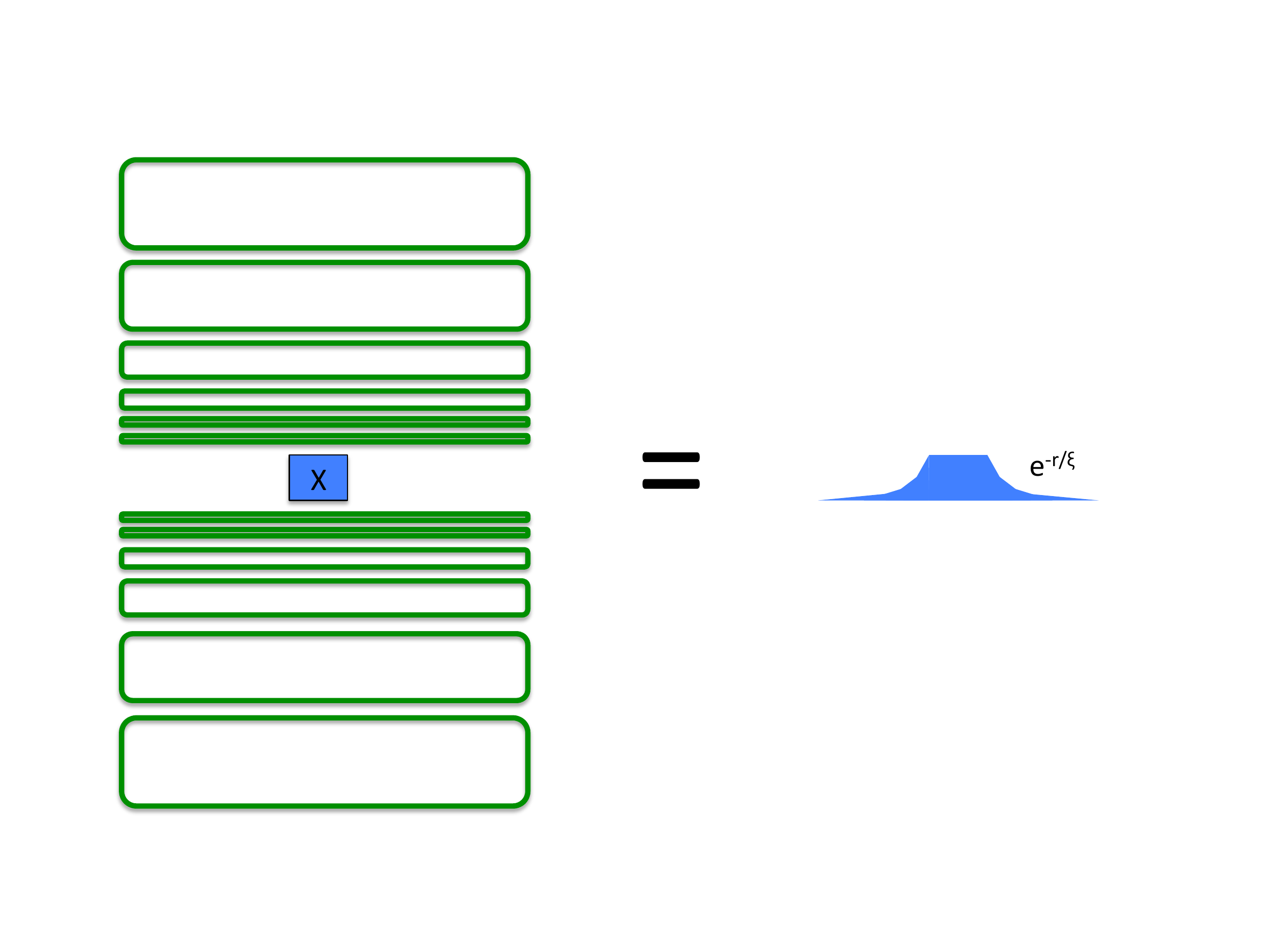}

Consider the implications of this result for the eigenvalues $H_{diag}$. 
\begin{align*}
H_{diag} =UHU^{\dag}=\gamma\sum_i(\sigma^x_i)_{tailed} +\\+\sum_{i}h_i (\sigma^z_i)_{tailed} + J_i (\sigma^z_i \sigma^z_{i+1})_{tailed}
\end{align*}
On the other hand, since $H_{diag}$ is a diagonal matrix:
\begin{equation}
H_{diag} =f(\sigma_z) = \sum_{s\subset \{1\dots L\}} J_s \prod_{i\in s} \sigma_i^z.
\end{equation}
The bound on $J_s$ can be directly inferred from the general bound on the telescopic sum (Theorem \ref{thepain} in the Appendix):
\begin{corollary}
\begin{equation}
 J_s \leq 12 \alpha^{d(s)}
\end{equation}
where $\alpha = 24c'\chi^ {\frac{56}{225}} $,  $d(s)$ is \emph{consecutive support} and no resonances are within it.
\end{corollary}

Using a time-averaged version of this (see Section \ref{tia}) as a starting point in an ensemble-averaged sense, Kim et al\cite{Isaac} derive Lieb-Robinson bound:
\begin{corollary}
\begin{align*}
\| [A,B(t)]\| \leq ate^{-bd_{A,B}}=\\ = at\alpha^{d_{A,B}}
\end{align*}
where $B(t) = e^{iHt} B e^{-iHt}$, $\alpha = 24c'\chi^ {\frac{56}{225}} $, $a=c\|A\|\|B\|$ and $c$ is a constant number.
\end{corollary}
We provide an independent derivation in Appendix A. Another known MBL effect of logarithmic growth of entanglement across the cut can also be proven\cite{Isaac}. We prove that $U$ is close to a finite-depth local circuit in Appendix \ref{absCarry}.

The benefit of taking Postulate \ref{Conj} as a starting point is that one gets area law essentially for free. Consider an eigenstate $|E_s\rangle= U^\dag|s\rangle$  where $|s\rangle$ is a product state of spins in z-basis and a cut at some location:
\begin{theorem}[Area law]
\begin{equation*}
S_{\text{cut}} = S(\rho_{left})< 650 c'\chi 
\end{equation*}
where $\rho_{left} = \textrm{tr}_{right} |E_s\rangle \langle E_s|$.
\end{theorem}
No other proofs of area law in an eigenstate appeared in the literature to our knowledge. It was considered to be an obvious corollary of exponential decay of correlations, however Imbrie did not prove exactly the same decay of correlations as required for 1d area law.\cite{Horo} We show our proof in Appendix A, and obtain the Postulate \ref{Conj} from the construction by Imbrie\cite{Imbrie} in Appendix B.

We now interleave the circuit with {\emph{resonances}} ~--- locations where perturbation theory in $\gamma$ failed and a local non-perturbative rotation had to be applied.
\begin{definition} {\emph{(R-size)}}
At each step, the operators acting between $e^{A^{(k)}}$ are 
\begin{equation}
O_k =\prod_r O_{k,r},
\end{equation}
where $O_{k,r}$ are big unitary rotations supported on $<4.2L_k =4.2(15/8)^k$  consecutive spins.
\end{definition}
Let $R_k$ be the set of spins involved in $O_k$. Imbrie\cite{Imbrie} explicitly presents a following bound on how rare they are:
\begin{definition} {\emph{(R-density)}}
The density of step-$k$ resonances
\begin{equation}
P(i\in R_k) = \text{lim}_{L\to \infty}\frac{|R_k|}{L}\leq c\epsilon^{c'\text{log}^2L_k}
\end{equation}
where $\epsilon = \gamma^{0.05}$ and $c, c'$ are constant numbers.
\end{definition}
We postulate {\bf{1-4}} as a definition of a full $U$. Now entanglement bound strongly depends on the position of the cut!
\begin{corollary}
\begin{equation*}
\overline{S_{cut}} \leq   650 c'\chi + 4\epsilon\textrm{ln}2
\end{equation*}
and {\emph{max}}$(S_{cut})\sim \text{log} L$ 
\end{corollary}
Here $\overline{S_{cut}}$ is the average over different positions in the system (equivalently, different realizations of disorder), and the true finite size scaling is more complicated than log$L$ but is slower than any power. We note an incredibly small power in the resonance density bound $\epsilon = \gamma^{0.05}$. That means that Imbrie's construction as is cannot be applied for magnetic fields $\gamma\geq 10^{-6}$. MBL survives for much bigger magnetic fields, and this discrepancy is due to multiple overestimations needed for Imbrie's proof to got through. If one uses a more strict choice of which regions to call resonant, their density remains controlled roughly up to the transition to extended phase. We demonstrate it numerically in a companion paper\cite{JeniaGo,phdthesis}.

\section{Time-averaging}
\label{tia}
Consider the following transformation of a local operator $X$:
\begin{equation}
X' =\textrm{lim}_{T\to \infty}1/T \int_0^T e^{iHt} X e^{-iHt} dt;  
\end{equation}
this kills off-diagonal elements in the eigenbasis of $H$, so commutation with $H$ is enforced.
Time-averaging is used in Kim et al\cite{Isaac} as part of a definition of MBL. They show the following for time-averaging with an MBL Hamiltonian:
\begin{lemma}
Infinite time-averaging introduces exponential tails.
\begin{equation}
\|X'-X'_c\| \leq a(4\alpha)^{c}
\end{equation}
where $X'_c$ is the truncation of $X'$ to a region with collar $c$ around the support of $X$, $\alpha = 24c'\chi^ {\frac{56}{225}} $ and $a,Z$ are constant numbers.
\end{lemma}
\begin{proof} 
Following\cite{Isaac}, logical spin operators $U\sigma^\alpha_i U^{\dag}$ form an operator basis and obey the telescopic sum bound proven in Appendix {\bf{A}}. Coefficients of an operator $X$ in this basis are given by traces of products tr$XU\sigma^\alpha_i U_\dag$, which are bounded as $c'\alpha^{i-i_X}$ where $i_X$ is the nearest site in $X$. Only the coefficients of combinations of $U\sigma^z_i U^{\dag}$ survive the infinite time averaging. Then each term has a telescopic bound $c'\alpha^{c}$ where $c$ is the size of the collar. Two bounds combine to give a collar bound on $X'$. The $\alpha$ is only increased by a square of local dimension as there are many possible operator strings.
\end{proof} 

If we time-average a local patch of the Hamiltonian $H_a$ in $H= H_a +H_b$  for a finite-time, the commutation with $H$ is not fully enforced. Let's bound the commutator of the two parts:
\begin{align}
    [H_a',H_b'] =[H_a',H] = \frac{1}{T} \int_0^T e^{iHt} [H_a,H] e^{-iHt} dt  =\\
    =\frac{1}{iT} \int_0^T \frac{d}{dt}e^{iHt} H_a e^{-iHt} dt  =\frac{e^{iHT} H_a e^{-iHT} - H_a}{iT}\\
    \| [H_a',H_b']\| \leq \frac{2S}{T}, \label{comBound}
\end{align}
where S is the support of $H_a$. The spread of the local operator of any Hamiltonian can be bounded by the Lieb-Robinson bound. Let the light cone (where the bound on $\|[A(t),B]\| \sim 1$) be given by $x(t)$, therefore time-averaging for time $T$ should give an operator whose support outside the radius $x(T)+\Delta x$ can be bounded as $e^{-c\Delta x}$, where $c$ is a Hamiltonian-independent number for 2-local Hamiltonians (and can be strengthened even more for MBL Lieb-Robinson bound: $e^{-c\Delta x} = \alpha^{\Delta x}$ now $H$-dependant unlike the generic Lieb-Robinson bound). Moreover, the time-averaging can be done with the Hamiltonian truncated to an $x(T)+\Delta x$ collar of the local operator, with only exponentially small error:
\begin{equation}
     {\tilde{H}}_a' =1/T \int_0^T e^{iH_ct} H_a e^{-iH_ct} dt =   H_a' + \delta_a, \quad \|\delta_a\|\leq e^{-c\Delta x}.
\end{equation}
 Note that $\delta$ is a nonlocal interaction that's exponentially small in the distance as $e^{-cd}$. Now split the system in many regions $a_1,a_2 \dots$ So if we take $H_a$ to have support $S>x(T)+\Delta x$, the effects of the time-averaging could be truncated in such a way that ${\tilde{H}}_{a_i}'$ only overlap with their neighbors ${\tilde{H}}_{a_{i\pm1}}'$. So we have reduced our Hamiltonian to a nearest neighbor chain of approximately commuting terms (as in Bound \ref{comBound}), with an exponentially small error:
 \begin{equation}
      H = \sum_i {\tilde{H}}_{a_i}' + \delta_i.
 \end{equation}
Now let's take even-odd pairs ${\tilde{H}}_{a_{2k}}'+ {\tilde{H}}_{a_{2k+1}}'$. We would like to use a result by \citep{twomatonecup} that reduces approximately commuting operators to exactly commuting ones, by small adjustments to each:
\begin{align}
     {\tilde{H}}_{a_{2k}}' =A_{2k}+\epsilon_{2k}, \quad {\tilde{H}}_{a_{2k+1}}' =A_{2k+1}+\epsilon_{2k+1}, \\
     [A_{2k},A_{2k+1}] =0,
\end{align}
and the $\|\epsilon\|<(2S/T)^{1/6}$ and have the same support as the overlap of $a_{2k}'$ and $a_{2k+1}'$ (the supports of corresponding $H'$). Now we introduce the notation
\begin{align}
     h_k = A_{2k} + A_{2k-1},  \quad [h_k,h_{k+1}]=0, \\
     \quad H= \sum_k h_k +\sum_i (\epsilon_i + \delta_i).
\end{align}
We have reduced any local Hamiltonian on a chain to a nonlocal one with exponentially decaying long-range interaction, written in the form of commuting terms plus a small perturbation. We note that to repeat this in 2d one would need a version of Hastings' result for three matrices, and that one has a counterexample presented in his paper. 

For chains, one may think of applying Imbrie's construction to this form of the Hamiltonian - after step 1, Imbrie's Hamiltonian has exactly this form with long range but rapidly decaying interactions! Unfortunately, the bounds for a general 1d Hamiltonian do not work out in our favor: the number of levels in the block $S>x(T)+\Delta x$ grows as $2^S$, so the minimum level distances decrease as $2^{-S}$. However, the norm bound on $\epsilon$ is just a power $1/T$. For the standard Lieb-Robinson bound $x(T) =vT$ and $S>x(T)$, so the level distances are always smaller than the perturbation scale, which is the opposite to what is used in Imbrie's construction. However, in MBL with $x(T) \sim \text{ln}T$, the scales may just work out: $T \sim(1/\alpha)^{x(T)}$, so the norm of the perturbation goes as $\alpha^{\tilde{c}S}$ where $S$ is our choice of the support and $\tilde{c}$ is some constant. It decreases faster than $(1/2)^S$ for a sufficiently small $\alpha$. We have shown the following:

\begin{corollary}
If $H$ obeys MBL Lieb-Robinson bound
\begin{equation}
\| [A,B(t)]\| \leq at\alpha^{d_{A,B}},
\end{equation}
then the perturbation on top of a $S$-local commuting component of $H$
 $V=\sum_i V_i$ is sufficiently small for Imbrie's proof to go through:
\begin{equation}
\|V_i\| \leq c''\alpha^{c'S}
\end{equation}
where $c'',c', S$ are constant numbers.
\end{corollary}
In this way, any Hamiltonian that possesses the MBL Lieb-Robinson bound can be reduced to a form "diagonal + small perturbation", even if the z-basis is not apparent from the start.
\nocite{*}
\bibliography{main}

\appendix
\onecolumngrid
\section{Proofs of properties of Imbrie Circuit}
\subsection{Telescopic sum for Imbrie circuit}
\subsubsection{Small unitary expansion} We consider a generator $A^{(k)}$  of a unitary evolution $e^{A^{(k)}}$ (corresponding to step $k$) such that $A^{(k)}= \sum_i A_i^{(k)}$ where $\|A_i^{(k)}\| \leq c' \chi^{\textrm{ceil}L_{k-1}}$) and each $A_i$ is supported on floor$(\frac{8}{7}L_{k+1})$ consecutive sites. We would like to study the structure of a local operator $X$, evolved as $e^{A^{(k)}} X e^{-A^{(k)}}$. For the purposes of the proofs below, we use the term support $S(X)$ of any operator $X$ as the shortest interval of the chain containing all the sites where action of $X$ is distinct from identity. We prove

\begin{lemma}
\label{supercharge}
\begin{equation}
e^{A^{(k)}} X e^{-A^{(k)}}  = \sum_{j=0}^\infty X_j
\end{equation}
where  $ S(X_j) \leq S(X) +  j\Delta S $,  $\Delta S =2(\textrm{floor}\frac87 L_{k+1}-1)$  and 
\begin{equation}
\| X_j \|\leq 2^{j_0}(2c'\Delta S )^j \chi^{j\textrm{ceil}L_{k-1}} \|X\| 
\end{equation}
where  $j_0 = \textrm{ceil}\frac{S(X) +S(A) -2}{\Delta S}$ and $S(A) = \textrm{floor}\frac87 L_{k+1}$
\end{lemma}

So the operator $X$ acquires exponential tails after conjugating with $e^{A^{(k)}}$.

\begin{proof}

We drop the superscript $(k)$ in $A^{(k)}$ for convenience.
\begin{equation}
e^A X e^{-A} = \sum_n \frac{A^n}{n!} X \sum_m \frac{(-A)^m}{m!}
\end{equation}
collecting the terms that have the same power of $A$
\begin{equation}
=\sum_j \sum_{m=0}^k \frac{A^{j-m}}{(j-m)!} X \frac{(-A)^m}{m!}
\end{equation}
Compare this to the expansion of $j$'th order commutator
\begin{equation}
[A[A\dots[A,X]]]^{(j)} = \sum_{m=0}^j C_m^j A^{j-m} X (-A)^m
\end{equation}
where $C_m^j = \frac{j!}{(j-m)! m!}$
\begin{equation}
\sum_j \sum_{m=0}^j \frac{A^{j-m}}{(j-m)!} X \frac{(-A)^m}{m!} = \sum_j \frac{1}{j!}[A[A\dots[A,X]]]^{(j)}
\end{equation}
Now denote $X_j =  \frac{1}{j!}[A[A\dots[A,X]]]^{(j)}$. We found that
\begin{equation}
e^A X e^{-A}  = \sum_j X_j
\end{equation}

Recall that support of $X$ is $S(X)$ consecutive sites and the support of $A$ is $S(A) = \textrm{floor}\frac87 L_{k+1}$. A commutator $[A,O]$ for a local operator $O$ of support $S(O)$ consecutive sites is nonzero for only the terms in $A = \sum_i A_i$ that overlap with $O$. The index $i$ goes over every site, so there are $S(O) + S(A) -2$ such terms. The support of $[A,O]$ becomes $S(O) + 2S(A)-2$.  The support of $X_j$ which is the $j$'th order of the commutator just grows linearly:
\begin{equation}
S(X_j) =S(X) + j(2S(A)-2) = S(X) + j\Delta S
\end{equation}
where we denoted $\Delta S = 2S(A) -2$. The number of terms of $A$ that have nonzero commutator at $j$'th step is given by $S(O) +S(A)-2$ for $O=X_{j-1}$:
\begin{equation}
n_j = S(X_{j-1}) + S(A) -2 = S(X) +j\Delta S +S(A) -2
\end{equation}
For the norm bound on $X_j$, a prefactor before the $j$'th power of $\|A\|$ is $1/j!$ and a product of all $n_j$'s:
\begin{equation}
 \|X_j\|\leq \frac{1}{j!} \left(\prod_{i=1}^j n_i \right)\|A\|^j \|X\|
\end{equation}
Let us bound the prefactor:
\begin{align}
 \frac{1}{j!} \prod_{i=1}^j n_i = \frac{1}{j!} \prod_{i=1}^j (S(X) +S(A) -2+i\Delta S ) =\frac{\Delta S^j}{j!} \prod_{i=1}^j \left(\frac{S(X) +S(A) -2}{\Delta S}+i \right) \leq \\
 \leq \frac{\Delta S^j}{j!} \prod_{i=1}^j \left(\textrm{ceil}\frac{S(X) +S(A) -2}{\Delta S}+i \right) =  \frac{\Delta S^j}{j!} \prod_{i=1}^j (j_0+i ) = \Delta S^j \frac{(j+j_0)!}{j!j_0!} = \Delta S^j C_{j_0}^{j+j_0} 
\end{align}
where we denoted $j_0 = \textrm{ceil}\frac{S(X) +S(A) -2}{\Delta S}$. Finally we use $C_{j_0}^{j+j_0}  \leq 2^{j+j_0}$ to arrive at
\begin{equation}
 \frac{1}{j!} \prod_{i=1}^j n_i \leq (2\Delta S)^j 2^{j_0}
\end{equation}
Combining that with the bound on $\|A\|$, we get the desired form:
\begin{equation}
\| X_j \|\leq 2^{j_0}(2c'\Delta S )^j \chi^{j\textrm{ceil}L_{k-1}} \|X\| 
\end{equation}

\end{proof}

\paragraph{Aside: comparison with Lieb-Robinson Bound}
We have found that if one conjugates a local operator $X$ with $S(A)$-locally generated 1d evolution $e^A$ for small time $t\leq c'\chi^{\textrm{ceil}L_{k-1}}$, operator $X$ acquires exponential tails.

It is a much stronger result than if we apply Lieb-Robinson Bound for small $\chi$ for propagation of tails beyond the collar, but becomes trivial for $\chi=1$. Here's this LRB estimate for comparison - it still gives nontrivial answers even for big $\chi$:
 \begin{equation}
\sum_{j=c}^\infty\|X_j \|\leq  \textrm{exp}\{-(ac -  S(A)c'\chi^{\textrm{ceil}L_{k-1}}\} \|X\|
\end{equation}
Here $a$ is just some number $<20$ much smaller than $\textrm{ceil}L_{k-1}/(-\textrm{ln}\chi)$.
We shall use the stronger first bound as in our case $\chi\ll 1$.

\subsubsection{Infinite product of unitaries} Now we'd like to use our bound to estimate precision of collar approximations to:
\begin{equation}
\prod_{k=1} e^{A^{(k)}} X \prod_{k'=1} e^{-A^{(k')}}
\end{equation}
 as in the total Imbrie circuit without resonances. The order in the product will turn out to be unimportant for our bound. From the point of view of the physical meaning, if smaller $k$ act on the operator $X$ first, then $X$ is a physical operator and evolved $X$ represents the action on "logical" spins, or integrals of motion. In particular, for the diagonalization of the Hamiltonian, smaller $k$ act on the physical Hamiltonian first. If smaller $k$ act last, then $X$ is a logical operator, and its evolved form will be the representation of such operator on physical spins.

 We now present and prove the bound:
 
\begin{theorem}
\label{thepain}
\begin{equation}
 \prod_{k=1}^{\infty} e^{A^{(k)}} X \prod_{k'=1}^\infty e^{-A^{(k')}} = \sum_{j=0}^{\infty} X_j, \quad S(X_j) =S(X) +2j, \quad \|Y_n\| \leq 6\cdot 2^{0.5S(X)}\alpha^j\| X\|
\end{equation}
where $\alpha = 24c'\chi^ {\frac{56}{225}} $ and the order in the product is unimportant.

 If we add a collar of $c$ sites to the support of $X$, and take $A_{c}^{(k)}$ to be the sum of terms fully within the collared region, then 
\begin{equation}
\left\|\prod_{k=1}^{\infty} e^{A^{(k)}} X \prod_{k'=1}^\infty e^{-A^{(k')}} - \prod_{k=1}^{\infty} e^{A^{(k)}_c} X \prod_{k'=1}^{\infty}  e^{-A^{(k')}_c}\right\| \leq 24\cdot 2^{0.5S(X)}\alpha^{c+1}\|X\|
\end{equation}
 $\alpha\leq 0.5$ is required for the second bound to hold.
\end{theorem}

\begin{proof}

 We start from looking at individual rotation, that was discussed above:
 \begin{equation}
U_k X U_k^\dag  = \sum_{j_k} X_{j_k}, ~ ~  S(X_{j_k})  =  S(X) +2j_k\textrm{floor}(\frac87 L_{k+1}-1),  ~ ~  \| X_{j_k} \|\leq 2^{j_0}(2c'\Delta S )^j \chi^{j\textrm{ceil}L_{k-1}} \|X\|  
\end{equation}
 If there are two rotations, we get:
 \begin{align}
U_{k'} U_k X U_k^\dag U_{k'}^\dag = \sum_{j_k,j_k'} X_{j_k,j_k'}, \\  S( X_{j_k,j_k'})  =  2j_k\textrm{floor}\left(\frac87 L_{k+1}-1\right) + 2j_k'\textrm{floor}\left(\frac87 L_{k'+1}-1\right),  \\ \| X_{j_k,j_k'} \|\leq 2^{j_{0,k} + j_{0,k'}(j_k)}(2c'\Delta S_k )^{j_k}(2c'\Delta S_k' )^{j_k'} \chi^{j_k\textrm{ceil}L_{k-1} +j_k'\textrm{ceil}L_{k'-1} } \|X\|  
\end{align}
we see that the order of unitaries indeed does not matter except for the first term, that depends on the support of previous terms. We will see in paragraph \ref{geomm} how that term will remain bounded for every $X_{j_k,j_k'}$ s.t. $S(X_{j_k,j_k'}) =S(X) +2j$ regardless of the order. For an infinite product, the components of evolved $X$ are indexed by a string of integers $\{j_k\}$ for all $k=1,2,3\dots$: 
\begin{align}
\prod_{k=1} e^{A_k} X \prod_{k'=1} e^{-A_{k'}} = \sum_{\{j_k\}} X_{\{j_k\}}, \\  S(X_{\{j_k\}})  = S(X)+  \sum_{k}2j_k\textrm{floor}\left(\frac87 L_{k+1}-1\right) ,  \\  \| X_{\{j_k\}} \|\leq C(\{j_k\},\textrm{order}) \left(\prod_{k} (2c'\Delta S_k )^{j_k}\right) \chi^{\sum_{k}j_k\textrm{ceil}L_{k-1}} \|X\| \\
C(\{j_k\},\textrm{order}) =2^{j_{0,k(1)} + j_{0,k(2)}(j_{k(1)}) +\dots}
\end{align}
Terms $X_{\{j_k\}} $ whose resulting supports $S( X_{\{j_k\}})=S(X) +2j$ ($j\in {Z}$)  are combined into $X_j$:
\begin{equation}
X_j = \sum_{\{j_k\}|S( X_{\{j_k\}})=S(X) +2j}X_{\{j_k\}}
\end{equation}
We will provide a uniform bound in terms of $j$ on the norm of each $X_{\{j_k\}} $ in $X_j$, and also argue that the total number of terms is only exponential in $j$.

One of the terms contributing to$\| X_{\{m_k\}|j} \|$ is  
\begin{equation}
 A =\prod_{k}(2c'\Delta S_k )^{j_k}
\end{equation}
We know that $\Delta S_k = 2\textrm{floor}(\frac87 L_{k+1}-1)$. We will leave one of the powers of two to the $C(\{j_k\},\textrm{order})$, and bound a 
\begin{equation}
 A =\prod_{k} \left(2c'\textrm{floor}\left(\frac87 L_{k+1}-1\right)\right)^{j_k} \leq \prod_{k} \left(\frac{16}{7}c' L_{k+1}\right)^{j_k}
\end{equation}
The collar is increased in steps of $S(A_k)-1 =$floor$\frac87 L_{k+1} -1$ on each side, so 
\begin{equation}
j= \sum_k (S(A_k)-1)j_k
\end{equation}
We can collect two sums over $j_k$ in our bound:
\begin{equation}
 A \leq \left(\frac{16}{7}c'\right)^{\sum_k j_k} \left(\frac{15}{8}\right)^{\sum_k (k+1)j_k}
\end{equation}
Note that
\begin{equation}
 \sum_k j_k \leq \sum_k (k+1)j_k \leq \sum_k (S(A_k)-1)j_k = j
\end{equation}
so assuming $c'>1$ (or else it can be replaced by $1$ in the original bound):
\begin{equation}
 A \leq \left(\frac{30}{7}c' \right)^{j}
\end{equation}
That's a good enough bound for our purposes. Now let's turn our attention to powers of $\chi$:

\begin{equation}
 \chi^{\sum_{k}j_k\textrm{ceil}L_{k-1}} \end{equation}
We'd like to bound it by a $j$-dependent expression as well. $\chi<1$, so we need to bound the sums in the opposite direction:
\begin{align}
 \sum_{k}\textrm{ceil}L_{k-1}j_k \geq \sum_{k}L_{k-1}j_k  = \left(\frac{8}{15}\right)^2\frac{7}{8}\sum_{k}\frac{8}{7}L_{k+1}j_k \geq \\\geq \left(\frac{8}{15}\right)^2\frac{7}{8}\sum_{k}\textrm{floor}(\frac{8}{7}L_{k+1}-1)j_k = \left(\frac{8}{15}\right)^2\frac{7}{8} j
\end{align}
So the resulting power of $\chi$ is at least
\begin{equation}
 \chi^{\sum_{k}j_k\textrm{ceil}L_{k-1}} \leq \chi^ {\left(\frac{8}{15}\right)^2\frac{7}{8} j}
\end{equation}
Finally, we deal with the order-dependent factor:

\paragraph{$2^{j+j_0}$ factor}\label{geomm}
The factor $2^{j_0}$ might be problematic for the proofs involving the full circuit, as $X$ is oftentimes a big object, and the earlier steps expand it even bigger.

For the proof above, we needed the telescopic sum bound:
\begin{equation}
\to\| X_{j_k} \|\leq 2^{{j_k}+j_{0,k}}(c'\Delta S )^{j_k} \chi^{j_k\textrm{ceil}L_{k-1}} \|X\| 
\end{equation}
For convenience, we pulled out $2^{j}$. All factors except the first are independent on the multiplication order and their contributions were bounded above. We will show that the product of $2^{{j_k}+j_{0,k}}$ over all steps is bounded by $M^{j}$ for some constant $M$. The bound is not multiplication order-dependent.

The contribution of terms $2^{{j_k}+j_{0,k}}$ to a single representative of the $X_j$ is:
\begin{equation}
 \prod_{k=1}^{k_{max}} 2^{j_k+j_{0,k}} = 2^{\sum_{k=1}^{k_{max}}j_k+j_{0,k}}
\end{equation}
 Now we need to bound the sum:
\begin{equation}
 \sum_{k=1}^{k_{max}}j_k+j_{0,k} =  \sum_{k=1}^{k_{max}}j_k+\textrm{ceil} \frac{S(X_{k-1}) + S(A_k)-2}{2S(A_k) -2}
\end{equation}
where we have introduced $X_{k-1}$ as the term that we took from the previous step. The sequence of supports $S(X)=S(X_0),S(X_1)\dots S(X_{k_{max}})=2j + S(X_0) $ should be monotonic. $S(A_k) =$floor$\frac87 L_{k+1}$. Finally, $j_k$ is the number of terms in the commutator we selected at step $k$, so $S(X_k) = S(X_{k-1}) + j_k(2S(A_k)-2)$, so $j_k = (S(X_k) - S(X_{k-1})) /(2S(A_k)-2) $
\begin{equation}
 \sum_{k=1}^{k_{max}}j_k+j_{0,k} \leq  \sum_{k=1}^{k_{max}} \frac{S(X_{k}) + S(A_k)-2}{2S(A_k) -2} +1 \leq \frac32 k_{max} +(S(X_{k_{max}})-1) \sum_{k=1}^{k_{max}} \frac{1}{2S(A_k) -2} 
\end{equation}
The series at the end is bounded:
\begin{equation}
 \sum_{k=1}^{k_{max}} \frac{1}{2S(A_k) -2}  = \frac{1}{2}\sum_{k=1}^{k_{max}}\frac{1}{\textrm{floor}\frac{8}{7}L_{k+1} -1} \leq \frac{1}{2}\sum_{k=1}^{\infty}\frac{1}{\frac{8}{7}L_{k+1} -2}\leq \frac{1}{2}\sum_{k=1}^{\infty}\frac{1}{\frac{8}{7}L_{k} } = \frac{7}{30} \frac{1}{1-\frac{8}{15}}= \frac12
\end{equation}
So we arrive at a bound:
\begin{equation}
 \sum_{k=1}^{k_{max}}j_k+j_{0,k} \leq  j + \frac{3 k_{max}  + S(X_0) -1}{2} 
\end{equation}
we note that $k_{max}\sim $ln$j$, so we indeed only get an extra power of $j$, which will be suppressed by a small enough $\chi$. Let's recall the specific dependence of $k_{max}$ on $j$. $X_0=X$, then $X_1$ is everything within a one site collar, $X_j$ - within a $j$-site long collar. The collar is increased in steps of $S(A_k)-1 =$floor$\frac87 L_{k+1} -1$ on each side, so floor$\frac87 L_{k_{max}+1} -1\leq j$:
\begin{align}
 k_{max} = \textrm{max}_k \textrm{ s.t. }\textrm{floor}\frac87 L_{k+1} -1\leq j\\
 \frac87 L_{k_{max}+1} -2\leq j, \quad 
 (15/8)^{k_{max}}=L_{k_{max}} \leq\frac{7}{15}(j+2)  \\
 k_{max}\textrm{ln}(15/8)\leq \textrm{ln} (j+2) + \textrm{ln}(7/15) \quad k_{max}\leq \frac{\textrm{ln} ((7/15)(j+2))}{\textrm{ln}(15/8)} \leq  \frac{ (7/15)(j+2)}{\textrm{ln}(15/8)}
\end{align}
since $\frac{7}{15}(j+2)>1$. So our previous bound becomes:
\begin{equation}
 \sum_{k=1}^{k_{max}}j_k+j_{0,k} \leq  \frac{\left(2 + \frac{7}{5\textrm{ln}(15/8)}\right) j  + \frac{14}{5\textrm{ln}(15/8)}-1 +S(X_0)}{2} \leq 2.2 j + 1.8+0.5 S(X_0)
\end{equation}
And the prefactor is bounded as:
\begin{equation}
 \prod_{k=1}^{k_{max}} 2^{j_k+j_{0,k}}\leq 2^{2.2 j + 1.8+0.5 S(X_0)} = (2^{2.2})^j 2^{1.8 +0.5S(X_0)}
\end{equation}

Now putting it all together, we've learned that
\begin{equation}
 \|X_{\{j_k\}}\| \leq (2^{2.2})^j 2^{1.8 +0.5S(X_0)}  \left(\frac{30}{7}c' \right)^{j} \chi^ {\left(\frac{8}{15}\right)^2\frac{7}{8} j} \|X\|
\end{equation}
for $X_{\{j_k\}} \in X_j$. To bound $X_j = \sum_{\{j_k\}|j}X_j$ we need to upper bound the number of terms in the sum.

\paragraph{Number of terms}

We need to estimate the number of terms with collar $j$ that are generated in our telescopic sums $ \sum_{\{j_k\}} X_{\{j_k\}}$. First note that terms with $j_k\ne 0$ for at least one $k>k_{max}$ do not contribute. We're bounding the number of ways 
\begin{equation}
 \sum_k j_k \left(\textrm{floor} \frac87 L_{k+1}-1\right) =j
\end{equation}
can be achieved. Note that for each $j_k$
\begin{equation}
 j_k \left(\textrm{floor} \frac87 L_{k+1}-1\right) \leq j
\end{equation}
So there are at most $j/\left(\textrm{floor} \frac87 L_{k+1}-1\right)$ choices for each $j_k$ up to $k_{max}$. This gives an upper bound on the number of terms:
\begin{equation}
\sum_{\{j_k\}|j}1 \leq \prod_{k=1}^{k_{max}}\frac{j}{\textrm{floor} \frac87 L_{k+1}-1} = \frac{j^{k_{max}}}{\prod_{k=1}^{k_{max}}\textrm{floor} \frac87 L_{k+1}-1}
\end{equation}
For the upper bound on $ \sum_{\{j_k\}|j}1$, we need to lower bound the denominator:
\begin{equation}
 \prod_{k=1}^{k_{max}}\textrm{floor} \frac87 L_{k+1}-1 \geq  \prod_{k=1}^{k_{max}} L_{k} = \left(\frac{15}{8}\right)^{\sum_{k=1}^{k_{max}}k} = \left(\frac{15}{8}\right)^{\frac{k_{max}(k_{max}+1)}{2}}
\end{equation}
We obtain the bound on the number of terms as:
\begin{equation}
 \sum_{\{j_k\}|j}1 \leq j^{k_{max}}\left(\frac{8}{15}\right)^{\frac{k_{max}(k_{max}+1)}{2}}
\end{equation}
Note that we will need the opposite inequality on $k_{max}$:
\begin{align}
 k_{max} = \textrm{max}_k \textrm{ s.t. }\textrm{floor}\frac87 L_{k+1} -1\leq j\\
 \textrm{floor} \frac87 L_{k_{max}+2}-1 >j\\
 L_{k_{max}+2} > \frac78 j\\
\left(\frac{15}{8}\right)^{(k_{max}+1)} > \frac{7}{15}j
\end{align}
So
\begin{equation}
 \sum_{\{j_k\}|j}1 \leq j^{k_{max}}\left(\frac{7}{15}j\right)^{-\frac{k_{max}}{2}} = \left(\frac{15}{7}j\right)^{\frac{k_{max}}{2}}
\end{equation}
Using the inequality $ k_{max}\leq \frac{\textrm{ln} ((7/15)(j+2))}{\textrm{ln}(15/8)} $ derived before, we obtain:
\begin{equation}
 \sum_{\{j_k\}|j}1 \leq \textrm{exp}\left(\textrm{ln}\left(\frac{15}{7}j\right){\frac{k_{max}}{2}}\right) \leq  \textrm{exp}\left(\textrm{ln}\left(\frac{15}{7}j\right)\frac{\textrm{ln} ((7/15)(j+2))}{2\textrm{ln}(15/8)}\right)
\end{equation}
Expression on the right-hand side goes as exp(ln$^2j$), so it is clearly bounded by some power of $e^j$. To find the specific factor, we bound
\begin{align}
 \textrm{ln}\sum_{\{j_k\}|j}1 \leq\textrm{ln}\left(\frac{15}{7}j\right)\frac{\textrm{ln} ((7/15)(j+2))}{2\textrm{ln}(15/8)} = \frac{\textrm{ln}j \textrm{ln}(j+2) + \textrm{ln}\frac{15}{7}\textrm{ln}\left(1 +\frac{2}{j}\right) -  \textrm{ln}^2\frac{15}{7}}{2\textrm{ln}(15/8)} \leq \\\leq \frac{\textrm{ln}^2(j+2)+ \textrm{ln}\frac{15}{7}\textrm{ln}3 -  \textrm{ln}^2\frac{15}{7}}{2\textrm{ln}(15/8)}
\end{align}
now we use the property of a function ln$^2 x\leq 0.6x$ for $x\geq 3$. As an exercise, let's prove it. First note that second derivative of ln$^2x$ changes sign at ln$x =1$, $x =e$  and is negative for $x>e$, in particular for $x\geq 3$. Next we calculate the values of ln$^2x$ and $0.6x$ as well as their derivatives at $x=3$, which allows us to prove that  ln$^2 x\leq 0.6x$ at least until $x=7$. Then we evalute the slopes at $x=7$ and see that the slope of ln$^2 x$ is $0.55<0.6$. Since the second derivative is always negative, the ln$^2 x\leq 0.6x$ will be satisfied for all $x>7$. A numerical check shows that the inequality can be tightened to ln$^2 x\leq 0.55x$, but we will use the rigorous $0.6$:
\begin{equation}
  \textrm{ln}\sum_{\{j_k\}|j}1 \leq \frac{0.6(j+2)+ \textrm{ln}\frac{15}{7}\textrm{ln}3 -  \textrm{ln}^2\frac{15}{7}}{2\textrm{ln}(15/8)} = \frac{0.3j+ 0.73}{\textrm{ln}(15/8)} 
\end{equation}
Thus
\begin{equation}
\sum_{\{j_k\}|j}1 \leq (15/8)^{0.3j+ 0.73} 
\end{equation}
Now we have all the ingredients for the bound on the $X_j$ in the telescopic sum:
\begin{align}
 \|X_{j}\|\leq \sum_{\{j_k\}|j}\|X_{\{j_k\}}\| \leq \textrm{max}_{\{j_k\}|j}\|X_{\{j_k\}}\| \sum_{\{j_k\}|j}1 \leq \\ \leq  (2^{2.2})^j 2^{1.8 +0.5S(X_0)}  \left(\frac{30}{7}c' \right)^{j} \chi^ {\left(\frac{8}{15}\right)^2\frac{7}{8} j}(15/8)^{0.3j+ 0.73} \leq 6\cdot 2^{0.5S(X)}(24c'\chi^ {\frac{56}{225}})^j\|X\| \label{ofSmall}
\end{align}

\paragraph{Final step} Note that with the bound telescopic sum we can derive a bound for collared $\tilde{X} =\prod_k e^{A_c^{(k)}} X \prod_{k'}e^{-A_c^{(k')}} $, where $c$ is the extra support on each side of $X$ that is allowed to appear in $\tilde{X}$. To that accord, $A_c^{(k)}$ for every $k$ is chosen to include only terms in $A$ that lie fully within a collar of $c$ around $X$. Clearly $A_c^{(k)}$ obeys the same bounds as $A^{(k)}$, so in particular the decomposition $\tilde{X} = \sum_j \tilde{X}_j$ obeys the bound (\ref{ofSmall}). We see that $\prod_k e^{A^{(k)}} X \prod_{k'}e^{-A^{(k')}}$ and $\tilde{X}$ only differ in terms starting from $X_{c+1}$ (because they involve $A_i^{(k)}$'s in regions where $\tilde{X}$ is different). So the difference is bounded by the sum of those:
\begin{equation}
\to \| \prod_k e^{A^{(k)}} X \prod_{k'}e^{-A^{(k')}} -\prod_k e^{A_c^{(k)}} X \prod_{k'}e^{-A_c^{(k')}} \|\leq  \sum_{j=c+1}^\infty \|X_j\| + \|{\tilde{X}}_j\| \leq 2\sum_{j=c+1}^\infty \|X_j\|
\end{equation}
The sum converges:
\begin{equation}
 \sum_{j=c+1}^\infty \|X_j\| \leq \sum_{j=c+1}^\infty 6\cdot 2^{0.5S(X)}(24c'\chi^ {\frac{56}{225}})^j\|X\| =  6\cdot 2^{0.5S(X)}\frac{(24c'\chi^ {\frac{56}{225}})^{c+1}}{1 -24c'\chi^ {\frac{56}{225}}}\|X\|
\end{equation}
for $24c'\chi^ {\frac{56}{225}} <1$ which requires very small $\chi \lesssim \frac{10^{-7}}{c'^4}$ due to the untight bounds used in the derivation. Note that the original Imbrie's construction also required at least $\chi^{1/20} <0.5$ which results in $\chi\lesssim 10^{-6}$ ~--- it's not new for such proofs to accumulate large numbers. We can bound the denominator by $2$ for  $24c'\chi^ {\frac{56}{225}} <0.5$:
The sum converges:
\begin{equation}
 \sum_{j=c+1}^\infty \|X_j\| \leq 12\cdot 2^{0.5S(X)}(24c'\chi^ {\frac{56}{225}})^{c+1}\|X\|
\end{equation}
which results in the bound for the collared approximation:
\begin{equation}
\left\| \prod_k e^{A^{(k)}} X \prod_{k'}e^{-A^{(k')}} -\prod_k e^{A_c^{(k)}} X \prod_{k'}e^{-A_c^{(k')}} \right\|\leq  24\cdot 2^{0.5S(X)}(24c'\chi^ {\frac{56}{225}})^{c+1}\|X\|
\end{equation}

\end{proof}

\subsection{Circuit approximation}
\label{absCarry}

We would like to have a FDL circuit approximation of a unitary $e^A$. 
Consider a simple case of $S(A)=2$ first: $A = \sum_{i} A_{i,i+1}$. The Trotter approximation in $N$ steps is defined as:
\begin{equation}
U_T =(e^{\frac{1}{N}A_{even}}e^{\frac{1}{N}A_{odd}})^N
\end{equation}
where $A_{even} = \sum_{i ~ even} A_{i,i+1}$. The error is given by the next term in BCH formula while it remains small:
\begin{equation}
e^A = U_T + O( \|[A_{even},A_{odd}]\|/N)=U_T + O(L \|A_{i,i+1}\|^2/N)
\end{equation}
 Note that the error over whole macroscopic system diverges with size. However, over a small patch the error is small. Let's see how the divergence arises. Here we use big $O(x)$ notation in a sense that $\|e^A -U_T\|\leq Mx $ for some range of $x$ around $0$. To get the actual value of constant $M$ in the bound, one way is to expand the following combination of unitaries $e^{\epsilon}, e^{\delta}$:
 \begin{equation}
  e^{\epsilon} e^{\delta} -e^{\epsilon + \delta } = \frac{[\epsilon,\delta]}{2} + \sum_{j=3}^\infty \frac{\epsilon^j+\delta^j - (\epsilon+\delta)^j}{j!} + \epsilon\sum_{j=2}^\infty \frac{\delta^j }{j!}+ \delta\sum_{j=2}^\infty \frac{\epsilon^j }{j!} + \frac{\epsilon}{2}\sum_{j=1}^\infty \frac{\delta^j }{j!}+ \frac{\delta}{2}\sum_{j=1}^\infty \frac{\epsilon^j }{j!}
 \end{equation}
 Taking the norm of the above expression, one arrives to multiple series like $\sum_{j=2}^\infty \frac{\|\epsilon\|^j }{j!} = e^{\|\epsilon\|}-1-\epsilon$. To upper bound this function one uses a version of Taylor expansion with a mean-value form of the remained:
 \begin{equation}
  e^{x} = 1+ x + \frac{e^{\tilde{x}}x^2}{2} \quad \textrm{where} \quad \tilde{x}\leq x
 \end{equation}
With that we will be able to derive 
  \begin{equation}
  e^{\epsilon} e^{\delta} -e^{\epsilon + \delta } = \gamma,\quad \|\gamma\| \leq M \|[\epsilon,\delta]\| +O(\|\epsilon\|^3,\|\delta\|^3, \|\epsilon\|^2\|\delta\|,\|\epsilon\|\|\delta\|^2)
 \end{equation}
And the number in front of every term can be obtained explicitly. In fact, the higher order terms are also combinations of commutators ($e^{\epsilon}e^{\delta}$ is an element of a Lie group that can be defined free from matrix multiplication), so higher powers of $L$ only appear as $L^p\|A_i\|^{2p}$ or lower, but we will not need it for our purposes here. We now plug in $e^{\frac{1}{N}A_{even}}e^{\frac{1}{N}A_{odd}} - e^{\frac{1}{N}A} = \gamma$ where $\gamma$ is bounded as above into the Trotter formula:
\begin{equation}
\|e^A -U_T\| = (1+\|\gamma\|)^N -1 \leq e^{\|\gamma\| N}-1
\end{equation}
where we have used that all the coefficients in the Taylor expansion of $e^{\|\gamma\| N}$ are strictly bigger than ones in $(1+\|\gamma\|)^N$. Now we can again use the mean-value form of the remainder of the Taylor expansion to bound the above with explicit constant as long as $e^{\|\gamma\| N}$ is not too big:
\begin{equation}
\|e^A -U_T\|\leq O(\|[A_{even},A_{odd}]\|/N)  = O(L \|A_{i,i+1}\|^2/N)
\end{equation}
We will stick with the big O notation for convenience. Let's prove that for an operator $X$ supported over $S(X)$ sites the conjugation with $e^A$ is simulated by our Trotter approximation with system-size independent error:
\begin{lemma}
\label{bozu}
for $A= \sum_i A_{i,i+1}$
\begin{equation}
e^A X e^{-A} = U_T X U_T^\dag + O \left( ||X|| \frac{(S(X)+4N)}{N}||A_{i,i+1}||^2 \right)
\end{equation}
\end{lemma}
We see that the system size dependence drops out, also $U_T X U_T^\dag$ is an operator supported on $S(X)+4N$ sites, and expressed by a tensor network of size $(S(X)+4N) \times 4N$. The terms outside the causal cone of $X$ just cancel.
\begin{proof}
First we use the bounds like those in Lemma \ref{supercharge} for $e^A X e^{-A}$:
\begin{equation}
 e^A X e^{-A} =\sum_j X_j, \quad S(X_j) =S(X) +2j, \quad \|X_j\|\leq 2^{0.5S(X)} (4\|A_{i,i+1}\|)^j \|X\|
\end{equation}
Then we use the approximation by collared evolution like the one in Theorem \ref{thepain} (the final part of the proof) 
\begin{equation}
\|e^A X e^{-A} - e^{A_c} X e^{-{A_c}}\| \leq 2 \sum_{j=c+1}^\infty \|X_j\| \leq  4\cdot2^{0.5S(X)} (4\|A_{i,i+1}\|)^{c+1} \|X\|
\end{equation}
Here $c$ is the number of sites we are to include in the "collar" of $X$. $A_c$ are the terms in $A$ supported inside the collared region.  
The problem reduced to finding an FDL approximation to $ e^{A_c}$ supported on $S(X) + 2c$ - a finite number of spins. For that purpose, Trotter approximation is used. 
The total precision defined as: 
\begin{equation}
\|e^{A} X e^{-A} - e^{A_c}_T X e^{-{A_c}}_T \| \le   4\cdot2^{0.5S(X)} (4\|A_{i,i+1}\|)^{c+1}\|X\|  +O((2c+S(X)) \|A_{i,i+1}\|^2\|X\|/N)
\end{equation}
Note that if we were Trotter approximating $e^A$ instead of the collared one, we still get the same result $e^A_T X e^{-A}_T = e^{A_c}_T X e^{-{A_c}}_T $ by cancellations of the spare terms as long as $c\geq2 N$. So there exists an FDL circuit $e^A_T$ over the whole system that approximates $e^A$ locally in every patch.

Plugging in the collar as small as $c=2$ , we find the first term $\sim(4\|A_{i,i+1}\|)^{3}$ to be negligible (we assume $2^{0.5S(X)}$ is suppressed by powers of $\|A_{i,i+1}\|$ for big O bound to make sense). The only allowed $N$ for $c=2$ for the total circuit to reduce to collared circuit is $N=1$ The error is now:
\begin{equation}
\|e^{A} X e^{-A} - e^{A_c}_T X e^{-{A_c}}_T \| \le   O((4+S(X)) \|A_{i,i+1}\|^2\|X\|)
\end{equation}
More generally, we can take bigger $c=2N$, and the first term will always be suppressed as $\sim(4\|A_{i,i+1}\|)^{2N+1}$, which allows to neglect it compared to the second term:
\begin{equation}
\|e^{A} X e^{-A} - e^{A_c}_T X e^{-{A_c}}_T \| \le   O((4N+S(X)) \|A_{i,i+1}\|^2\|X\|/N)
\end{equation}
\end{proof}
Note that the precision never really gets better than $O(\|A_{i,i+1}\|^2\|X\|) $. That's concerning because our original collar $c=2$ approximation already had $\sim(4\|A_{i,i+1}\|)^{3}$ precision. So if we want to approximate that well, we need to resort to a higher order Trotter scheme (instead of $e^{\frac{1}{N}A_{even}}e^{\frac{1}{N}A_{odd}}$ it uses a more complicated step). Alternatively, since the dimension of the space we act upon ($d= 2^{S(X) +4}$ for $c=2$) is relatively small, we can instead use the Solovay-Kitaev theorem\cite{SK}. For this particular case it can be phrased as follows: to approximate a unitary on a $d$-dimensional Hilbert space to a precision $\epsilon$, one needs $n=O\left(d^a\textrm{ln}^b\frac{1}{\epsilon}\right)$ 2-qubit gates acting on nearest neighbors in 1d. Here $a,b$ are some constant numbers. So to achieve precision $\epsilon =\|A_{i,i+1}\|^{c+1}$ , one would need $n=O(c2^{ac})$ gates, or $O(2^{ac})$ depth. The number of gates in the $c$'th order of Trotter step also scales exponentially with $c$, giving the same $\|A_{i,i+1}\|^{c+1}$ precision. So whichever method wins depends on the constant $a$ and the scaling of the size of Trotter step.

Now let's apply this logic to the Imbrie circuit $\prod_{k=1}^\infty e^{A^{(k)}}$. We will see that as early as for approximation up to $k= 2$ the first order Trotter scheme becomes useless. Indeed, the $S(A_k)=$floor$(\frac{8}{7}L_{k+1})$, so $S(A_1) = 4$, $S(A_2) =7$, $S(A_3) =14$. The telescopic sum structure $X= X_j$ allows us to drop higher order steps, since step $j$ only uses the terms with $S(A_k)-1\leq j$, and make an error given by Theorem \ref{thepain}:
\begin{equation}
\left\| \prod_k e^{A^{(k)}} X \prod_{k'}e^{-A^{(k')}} -\prod_k^{k_{max}} e^{A_c^{(k)}} X \prod_{k'}^{k_{max}}e^{-A_c^{(k')}} \right\|\leq  24\cdot 2^{0.5S(X)}(24c'\chi^ {\frac{56}{225}})^{c+1}\|X\|
\end{equation}
For $c<6, ~ S(A_2)-1=6>c$, so we only include the 1st order term in the product and obtain the correct $X_1\dots X_6$:
\begin{equation}
\left\| \prod_k e^{A^{(k)}} X \prod_{k'}e^{-A^{(k')}} - e^{A_c^{(1)}} X e^{-A_c^{(1)}} \right\|\leq  24\cdot 2^{0.5S(X)}(24c'\chi^ {\frac{56}{225}})^{c+1}\|X\|
\end{equation}
Now, the $A^{(1)}_i$ has support $4$, but actually the first step of perturbation theory has support $3$, just our definition overestimates it for the sake of consistent form over all $k$. So it's enough to group spins into pairs before applying the Trotter: $A_{\mu,\mu+1} = A^{(1)}_i + A_{i+1}^{(1)}$ and $i=2\mu-1$. As a result we get an approximation $e^{-A_c^{(1)}}_T$ by 4-local gates bounded using Lemma \ref{bozu} as:
\begin{equation}
\|e^{A^{(1)}_c} X e^{-A^{(1)}_c} - e^{A_c^{(1)}}_T X e^{-{A_c^{(1)}}}_T \| \le   O((4N+S(X)) \|A_{\mu,\mu+1}\|^2\|X\|/N) \le O((4N+S(X)) (c'\chi)^2\|X\|/N)
\end{equation}
and due to the grouping the condition of one circuit for the whole system is $4N\leq c$. It can only be satisfied for $c\geq 4$, and $N=$floor$c/4=1$ always:
\begin{equation}
\left\| \prod_k e^{A^{(k)}} X \prod_{k'}e^{-A^{(k')}} - e^{A_c^{(1)}}_T X e^{-{A_c^{(1)}}}_T  \right\|\leq  24\cdot 2^{0.5S(X)}(24c'\chi^ {\frac{56}{225}})^{c+1}\|X\| +O((4+S(X)) (c'\chi)^2\|X\|)
\end{equation}
Comparing the powers, we note that ${\frac{56(c+1)}{225}}>2$ for $c\geq4$, so the first term can be omitted for sufficiently small $\chi$. We found that Trotter-approximating the first step of perturbation theory leads to an error
\begin{equation}
\left\| \prod_k e^{A^{(k)}} X \prod_{k'}e^{-A^{(k')}} - e^{A_c^{(1)}}_T X e^{-{A_c^{(1)}}}_T  \right\|\leq O((4+S(X)) (c'\chi)^2\|X\|) \label{erx}
\end{equation}
everywhere except for rare resonances. One can chose depth $2$ and $4$-local gates, or depth $3$ and $3$-local gates (the second application of Trotter bound leads to the same scaling of error). To go down to $2$-local gates, one may apply Solovay-Kitaev to the $3$-local gates, blowing up the depth by $O\left(8^a\textrm{ln}^b\frac{1}{\epsilon}\right) = O(|\textrm{ln}\chi|^b)$ factor. Note that even though for a uniform deep circuit the causal structure would make $e^{A_c^{(1)}}_T X e^{-{A_c^{(1)}}}_T$ depend on $O(|\textrm{ln}\chi|^b)$ collar, the long patches of gates arising from Solovay-Kitaev are arranged in a checkerboard pattern, so the true causal cone remains within $c$.

To conclude, let us study the relationship between a state $\prod_k e^{A^{(k)}} X|prod\rangle $ and $e^{A_c^{(1)}}_T |prod\rangle$. Any observable $X$ on $S(X)$ sites is approximated by measuring it with $e^A_T |0\rangle $ with error as in Eq. (\ref{erx}). The latter measurement is represented by a system-size independent tensor network due to the cancellations. In particular, we can measure a matrix element of the density matrix over those $S(X)$ sites. The difference in density matrices over $S(X)$ sites is bounded element-wise:
\begin{equation}
|(\rho -\rho_{appr})_{\alpha\beta}| \le   O((4+S(X)) (c'\chi)^2\|X\|)
\end{equation}
We can also show the operator norm bound:
\begin{equation}
\|\rho -\rho_{appr}\| \le   O((4+S(X)) (c'\chi)^2\|X\|)
\end{equation}
here $\rho = tr_{\overline{S(X)}}\prod_k e^{A^{(k)}} |prod\rangle \langle prod|\prod_{k'}e^{-A^{(k')}}$ and $\rho_{appr} = tr_{\overline{S(X)}}e^{A_c^{(1)}}_T |prod\rangle \langle prod| e^{-{A_c^{(1)}}}_T $. $|\prod \rangle$ is a product state in z-basis. To prove the norm bound, consider the eigenbasis of $\rho -\rho_{appr}$. We can make $X$ to be measuring matrix elements in this basis. In particular, the maximum eigenvalue of $\rho -\rho_{appr}$ will not be bigger than the difference $|\langle X \rangle - \langle X \rangle_{appr}|\le   O((4+S(X)) (c'\chi)^2\|X\|)$. Note that for the bound to be correct, the size for causal cone of $X$ should be $S(X) +2c$ and $c\geq 4$, and it does not improve the bound to have $c>5$.

To go to a better precision than $(c'\chi)^2$, specifically $\chi^p$, the depth required is exponential in $p$ until we allow some threshold size $S(X) +2c\sim p$ of gates in our circuit. Indeed, even for $A = \sum_i A_{i,i+1}$ above there isn't a subexponential way to achieve high precision using methods desribed so far. The threshold is reached when that $\sum_{j=1}^c X_j$ needed to approximate to precision $\chi^p$ is contained entirely within one gate, so we can do $\prod_k^{k_{max}} e^{A_c^{(k)}}$ in one step. This does not allow us to have a total circuit though - so methods described so far always require exponential in $p$ size of the circuit if one wants to have one circuit working for operator in any place of the system that reduces to a collared one via cancellations. There might be a method of getting a total circuit of $\sim p$ local gates and subexponential depth via splitting $e^{A} =(e^{A_L} \otimes e^{A_R}) e^{G_{LR}}$ where  $L,R$ are half-chains one the left and right of a specific cut. The $G_{LR}$ defined in this way does not commute with $A_L$, $A_R$, so the error above may be $\sim \|A_R\|$ - divergent with the system size. In fact it's not! The $G_{LR}$ has a collar bound/telescopic sum around the cut that can be found in \cite{EisertOsborne}, but we do not pursue the construction of precise subexponential depth circuit here.


\subsection{Area law}
We will prove the following:

\begin{corollary}
 The state $\prod_{k=1} e^{A^{(k)}}|prod\rangle$ obeys the area law: for $B$ - the half of the system, the entanglement entropy $S(\rho_B) \leq 650 c'\chi$. Here $|prod\rangle$ is a product state in the $z$-basis.
\end{corollary}
We will also discuss the tail of the distribution of entanglement entropies in the presence of resonances.

We turn to the paper \cite{LRBtop} which derives a bound on entanglement production $S(t) -S(0)$ under the evolution $e^{iHt}$. In our case, we apply the bound for every step using $iHt = A^{k}$. We set $t=1$. 

We use the bound that is given on the last page of \cite{LRBtop}. The interaction between two parts in $H =-iA^{k}$ can be written as a Shmidt decomposition $H_{LR}^{(k)} = \sum_n r_n J_L^n \otimes J_R^n$, where we have freedom to set $||J_L||=1$, $||J_R||=1$. The $A^{k}_i$ are bounded by $\|A_i^{(k)}\| \leq c' \chi^{\textrm{ceil}L_{k-1}}$ and there are $S(A^{(k)})-1=$floor$(\frac{8}{7}L_{k+1}-1)$ of them across the cut, thus the interaction is bounded as $\|H_{LR}^{(k)}\|\leq c'(S(A^{(k)})-1) \chi^{\textrm{ceil}L_{k-1}} $. We can decompose each $A^{k}_i$ across the cut into $4^{S(A^{k}_{i,L})}$ terms ~--- this is how much basis operators are there on one side of the cut. Each term in the decomposition is bounded as the whole operator that's decomposed $\|A_i^{(k)}\|$:
\begin{equation}
 \|r_n\|\leq c' \chi^{\textrm{ceil}L_{k-1}}
\end{equation}
There are a total of $\sum_{S(A^{k}_{i,L})=1}^{S(A^{(k)})/2}4^{S(A^{k}_{i,L})} \leq \frac43 2^{S(A^{(k)})}$ terms. The entangelement  produced over time $t=1$ will be 
\begin{equation}
S(t) -S(0) \leq  c^* \sum_i \|r_i\| \leq c^*\frac43 2^{S(A^{(k)})}  c' \chi^{\textrm{ceil}L_{k-1}}
\end{equation}
The coefficient $c^*\approx 1.9$ is universal and is found in \cite{CLV}. The total entanglement is the sum of entanglement produced by each order $k$ of the circuit. The total bound on entanglement is converging:
\begin{align}
\delta S_{tot} \leq \sum_k  c^*\frac43 c' 2^{\textrm{floor}(\frac{8}{7}L_{k+1})}\chi^{\textrm{ceil}L_{k-1}}  \leq\\\leq c^*\frac43 c'2^{\textrm{floor}(\frac{8}{7}L_{2})}\chi^{\textrm{ceil}L_{0}} \sum_{j=0}^{\infty}\left(2^{\textrm{floor}(\frac{8}{7}L_{3})-\textrm{floor}(\frac{8}{7}L_{2})}\chi^{\textrm{ceil}L_{1}-\textrm{ceil}L_{0}}\right)^j \leq \\ \leq \frac{ c^*\frac43 c'2^{7}\chi}{1 -2^{7}\chi} \leq  c^*\frac43 c'2^{8}\chi \leq 650 c'\chi
\end{align}
where we replaced the sum of exp$(L_k) = $exp$((15/8)^k))$ by a geometric series with a factor such that every term upper-bounds the corresponding term in the original sum, and $\chi<1/2^8$ has been used to get rid of the denominator. Thus the Imbrie states without resonances satisfy area law.

Every resonance is an extra local rotation over region of size $n_r$ that does not have to be small. If it happens that our cut passes through the resonance, we need to add the entanglement created by this rotation to the total entanglement bound. A rotation over $n_r$ sites can create $\Delta S(n_r) \leq n_r$ln$2$ entanglement. Naive expectation is $\leq (n_r/2)$ln$2$ which is true for a unitary acting on a whole system of size $n_r$. But in fact having a bigger system we can create two times more entanglement with the rotation on $n_r$ sites!  Indeed, the bound is saturated by a swap gate,  that turns a state with $n_r/2$ Bell pairs on every side of the cut into a state where every Bell pair crosses the cut.
We believe that this is an overestimation and only two levels are mixed at a typical resonance, producing $\sim$ln$2$ entanglement.

After $k_1$ when the first resonance appears across the cut, there may be other, bigger resonances containing the first one. Since a cut is placed randomly in the system, denote the probability that the biggest resonance across the cut has size $n_r$ by $P(n_r)$. It is exactly the same quantity as the probability of a site to be a part of a connected cluster of resonances of size $n_r$. Imbrie shows that this probability decays faster than any polynomial. Thus the expectation value
\begin{equation}
\sum_{n_r}n_rP(n_r) \leq \textrm{const}
\end{equation}
Moreover, from investigating the first step of the construction we know that the first term has $n_r=3$ and $P(n_r=3)=\epsilon$ (with some coefficient depending on the probability distribution of disorder, in our case $1$) and then it keeps decaying faster than any polynomial, so the leading contribution to the expectation value is the first term:
\begin{equation}
\sum_{n_r}n_rP(n_r) \leq 4\epsilon
\end{equation}
Here the first factor of $2$ is the bound on entanglement instead of $3$, as one of the sides has one site, and the second factor of $2$ absorbs the later terms in the sum, which can be done for sufficiently small $\epsilon$. Thus the total bound on the expectation value of the entanglement entropy is:
\begin{equation}
\overline{S_{tot}} \leq   650 c'\chi + 4\epsilon\textrm{ln}2
\end{equation}
We note that the first term $\gamma^{0.95}$ is smaller than the second term $\gamma^{0.05}$ ~--- so the main contribution to average entanglement is from rare distributions of disorder when resonances cross the cut. We also note that in full construction by Imbrie, the perturbative parts of the circuit $e^{A^{(k)}}$ get adjusted in the presence of resonances as the resonant spins are grouped into metaspins. The entanglement produced by resulting unitaries is still bounded, so we dropped this complication from our simplified circuit.

 Consider rare cuts that may have entanglement bigger than $2$ln$2$, by passing through the big resonant regions. Assume for simplicity exponential  distribution of entanglement as a discrete random variable $P(S=2x$ln$2) = \epsilon^{x}$  for $x$ the half of the size of the resonance. The biggest entanglement across the cut in a given system determined by:
\begin{align}
LP(S_{max}) = L\epsilon^{-S_{max}/2\textrm{ln}2}\approx 1 \\
S_{max} \approx  2\textrm{ln}2\frac{\textrm{ln} L}{(-\textrm{ln} \epsilon)}
\end{align}
To derive maximal entanglement more carefully one needs to use the true statistics of resonances from Imbrie's paper, summarized in Postulate {\bf{4}}, that results in a superpolinomially decaying $P(S)$ anyway. It only adjusts $\textrm{ln} L \to e^{\sqrt{\textrm{ln}L}}$ - a difference between the two will be almost impossible to detect.

\subsection{LRB}

Consider 
\begin{equation}
A(t) = e^{iHt}Ae^{-iHt} =U^\dag e^{iH_zt}UAU^\dag e^{-iH_zt}U = \sum_r U^\dag e^{iH_zt}A_r e^{-iH_zt}U
\end{equation}
where $H$ is the Hamiltonian of our spin chain and $U =\prod_{k=1} e^{A^{(k)}}$ is the unitary that diagonalizes it. According to Theorem \ref{thepain}, $A_r$ supported on $S(A) +2r$ sites and is bounded by $6\cdot 2^{0.5S(A)}\alpha^r\|A\|$ with $\alpha = 24c'\chi^ {\frac{56}{225}} $. Now the terms in $e^{iH_zt}$ that do not overlap with $A_r$ commute through, so
\begin{equation}
e^{iH_zt}A_r  e^{-iH_zt} = e^{iH_z^r t}A_r e^{-iH_z^r t}
\end{equation}
where $H_z^r = \sum_m H_z^{r+m} $ and the terms are bounded as $36\alpha^m$ (using $H=\sum_i H_i$, $\|H_i\| \leq 2+\gamma\leq3$ and $S(H_i) =2$). Now we find $m^*$ such that $\alpha^{m^*}t = 1$. Define
\begin{equation}
Q_{r+m^* +p} =e^{i\sum_m^{m^*+p}H_z^{r+m} t}A_r e^{-i\sum_m^{m^*+p}H_z^{r+m} t}- e^{i\sum_m^{m^*+p-1}H_z^{r+m} t}A_r e^{-i\sum_m^{m^*+p-1}H_z^{r+m} t}
\end{equation}
and $Q_{r+m^*} = e^{i\sum_m^{m^*}H_z^{r+m} t}A_r e^{-i\sum_m^{m^*}H_z^{r+m} t}$. The higher-order $Q$'s are bounded
\begin{equation}
\|Q_{r+m^* +p}\| \leq \|e^{iH_z^{r+m^* +p} t} - 1\| \|A_r\| \leq  \|A\|c^*\alpha^p
\end{equation}
Where $c^*$ is a new constant coefficient. We found that the telescopic sum for $e^{iH_z^r t}A_r e^{-iH_z^r t}$ has the form
\begin{equation}
e^{iH_z^r t}A_r e^{-iH_z^r t} = \sum_p Q_{r+m^* +p}, \quad \|Q_{r+m^* +p}\| \leq \|A_r\|c^*\alpha^p \label{quest}
\end{equation}
After plugging it in the decomposition of $A(t)$:
\begin{equation}
A(t) =\sum_r\sum_p U^\dag  Q_{r+m^* +p} e^{-iH_zt}U
\end{equation}
The last conjugation with the Imbrie circuit produces one more telescopic sum. Resumming all three telescopic sums, we get:
\begin{equation}
A(t) =\sum_q  Q_{m^*}', \quad \|Q_{m^* +q}\| \leq \|A_r\|\frac12 c^{**}\alpha^q 
\end{equation}
A new constant and a factor $\frac12$ was introduced for convenience. Now a local operator $B$ distance $d_{AB}$ from $A$ will only overlap with $m^*+q \geq d_{AB}$ terms in the telescopic sum. Resumming once more, we obtain:
\begin{equation}
\|[A(t),B]\| \le ||A|| ||B||c^{**}\alpha^{d_{A,B}-m^*} = ||A|| ||B||c^{**}\alpha^{d_{A,B}}t
\end{equation}

That concludes the proof.
\section{Interpretation of Imbrie's proof}
\subsection{Bound on generators}
\label{normToy}
 Here we discuss the bound on the local generator $A^{(k)}_i$ of the perturbative step $e^{\sum_i A^{(k)}_i}$ of Imbrie circuit. Here $i$ is the spatial index labeling sites of the chain, and $k$ is the order of the perturbation theory. Each $A^{(k)}_i$ is acting on floor$(\frac{15}{7}L_k)$ spins, where the radius of locality $L_k = (1\frac{7}{8})^k$ is growing with $k$. We have presented a bound
 \begin{equation}
 \|A^{(k)}_i\|  \leq c'\chi^{\textrm{ceil}L_{k-1}} \label{TarBo}
 \end{equation}
Imbrie works with a different decomposition: instead of $A^{(k)} =\sum_i A^{(k)}_i$ he uses $A^{(k)}=\sum_g A^{(k)}(g)$, where the sum iterator $g$ will be explained below. On p. 27 of the arXiv version of the paper\cite{Imbrie}, Imbrie presents a combinatorial argument that  $\sum_{g|i} A^{(k)}(g)$ of terms contributing to $A^{(k)}_i$ converges, and moreover the Bound (\ref{TarBo}) holds, however the Bound is not explicitly presented.
 
Let's define the path formalism used by Imbrie in the following toy examples to facilitate the understanding of the argument on p. 27.

\paragraph{Toy example 1.}
Consider an operator $J$ defined as:
\begin{equation}
J= \sum_g J(g)
\end{equation}
The sum is over paths $g = \left\{ i_1\dots i_{|g|}\right\}$ in the space of spin flips. Each $i_n$ is a site on a lattice, and the length of the string of those sites $g$ is denoted $|g|$. In the $\sum_g$, it is implied that $|g|>1$. There's one condition on the string $g$: each new spin flip has to be within $1$ lattice spacing from one of the old ones:
\begin{equation}
\forall n>1 ~ \exists n'<n : |i_n -i_{n'}|\leq 1
\end{equation}
The operators $J(g)$ are supported on the smallest consecutive region containing all spin flips in $g$ plus a 1 spin collar on the left and on the right. Consider the following way to bound the terms in $J$:
\begin{equation}
|J_{\sigma\sigma'}(g)|\leq \frac{\delta^{|g|}}{C(|g|)}
\end{equation}
where $C(x)$ is a combinatorial factor to be determined and $\delta \ll 1$. We'd like to see what $C(x)$ do we need to be able to translate this bound to a more traditional norm bound. Consider a decomposition of $J$ into "local" terms centered around sites of the lattice:
\begin{equation}
J = \sum_i J_i, \quad J_i = \sum_{g|i_1 =i}J(g)
\end{equation}
where $J_i$ contains only graphs starting at $i_1$. Let's show that a norm bound on $J_i$ follows from a bound on $J_{\sigma\sigma'}(g)$:
\begin{equation}
\|J_i\| \leq \sum_{g|i_1 =i}\|J(g)\| 
\end{equation}
Note that any operator $J(g)$ supported on $\cup g \cup c$ where $c$ is a 1-spin collar can be written in an operator basis with one basis vector per matrix element :
\begin{equation}
J(g) = \sum_{\sigma\sigma'}J_{\sigma\sigma'}(g)|\sigma\rangle\langle\sigma'|
\end{equation}
and since $J(g)$ acts as an identity outside $\cup g \cup c$, and doesn't flip $c$, the number of terms in the sum is $4^|\cup g|2^|c|\leq 2^{2|g| +c}$:
\begin{equation}
 \sum_{g|i_1 =i}\|J(g)\| \leq \sum_{g|i_1 =i} \sum_{\sigma\sigma'}J_{\sigma\sigma'}(g)\leq \sum_{g|i_1 =i}2^{2|g| +C}\frac{\delta^{|g|}}{C(|g|)}
\end{equation}
We can now split the sum into terms corresponding to each $|g|$:
\begin{equation}
\|J_i\| \leq\sum_{x =1}^{\infty}2^{2x +C}\frac{\delta^{x}}{C(x)}\sum_{g|i_1 =i, ~ |g|=x}1
\end{equation}
A simple upper bound on the number of elements in the sum $\sum_{g|i_1 =i, ~ |g|=x}1$ is if a flip $i_n$ can be chosen from $2n-1$ sites centered at $i_1$. For $|g|=x$ flips, it gives $(2x-1)!! = (2x-1)(2x-3)\dots$.
\begin{equation}
\|J_i\| \leq\sum_{x =1}^{\infty}2^{2x +C}\delta^{x}\frac{(2x-1)!!}{C(x)}
\end{equation}
So a choice of $C(x)=(2x-1)!!$ will cancel this factor. What remains is a geometric series which is bounded for $4\delta<1/2$:
\begin{equation}
\|J_i\| \leq 2^{C+1}(4\delta)
\end{equation}
we have arrived at a desired bound.

\paragraph{Toy example 2.} We will repeat the above construction with a minor modification: in the 
\begin{equation}
J= \sum_g J(g)
\end{equation}
the $\sum_g$ is now taken over $g$ such that $|g|>L_1$. The bound on $J_{\sigma\sigma'}(g)$ and $C(x) =(2x-1)!!$ are the chosen the same. All the steps go through in exactly the same way, until we arrive to:
\begin{equation}
\|J_i\| \leq\sum_{x =L_1}^{\infty}2^{2x +C}\delta^{x}\frac{(2x-1)!!}{C(x)} = \sum_{x =L_1}^{\infty}2^{2x +C}\delta^{x} \leq 2^{C+1}(4\delta)^{L_1} 
\end{equation}
we see that any level is bounded by the smallest allowed flip.

\paragraph{Toy example 3.} We will now introduce an extra layer of structure to match Imbrie's construction. Consider
\begin{equation}
J= \sum_G J(G)
\end{equation}
where $G =\left\{ g_1, g_2 \dots g_{|G|}\right\}$ a string of spin flip paths $g_n$ defined in examples before. The costraint on the paths in $G$ is now:
\begin{equation}
\forall n>1 ~ \exists n'<n : d(g_n,g_{n'})\leq 1
\end{equation}
where $d(g_1,g_2)$ is the distance in lattice spacings between the two nearest spins in $g_1,g_2$. 

The support of $J(G)$ is in the smallest consecutive region containing all spin flips of underlying paths of $G$ plus some constant collar $c$. Surprisingly, it is enough to bound
\begin{equation}
|J_{\sigma\sigma'}(G)|\leq \frac{1}{C(|G|)}\prod_{g\in G}\frac{\delta^{|g|}}{C(|g|)}
\end{equation}
where $C(x) = (2x-1)!!$ again. Indeed, for a "local" term
\begin{equation}
J_i = \sum_{G| i=i_1\in g_1 \in G}J(G)
\end{equation}
the norm bound
\begin{equation}
\|J_i\|\leq \sum_{G| i=i_1\in g_1 \in G}\|J(G)\|\leq \sum_{G| i=i_1\in g_1 \in G} \frac{4^{\sum |g|}2^c}{C(|G|)}\prod_{g\in G}\frac{\delta^{|g|}}{C(|g|)}
\end{equation}
The sum over $G$ can be split into sums over paths with fixed values of $|G|$ and all $|g_m|$:
\begin{equation}
\|J_i\|\leq   \sum_{|G|,|g_m|} \frac{4^{\sum |g_m|}2^c}{C(|G|)}\prod_{m}\frac{\delta^{|g_m|}}{C(|g_m|)}\sum_{G| i=i_1\in g_1 \in G, ~|G|,|g_m|} 1
\end{equation}
so we need to do the counting again. We first count the small graphs as before, and then count how many ways are there to place $|G|$ small graphs of sizes $|g_m|$ so that each new one is a neighbor to one of the old ones. We get the following bound:
\begin{align}
    \sum_{G| i=i_1\in g_1 \in G, ~|G|,|g_m|} 1 \leq \\ \leq \prod_{m} C(|g_m|) |g_1|(|g_1|+2|g_2|)(|g_1|+2|g_2| + 2|g_3|)\dots (2\sum_m |g_m| -|g_1|)
\end{align}
or if we plug it in:
\begin{align}
\|J_i\|\leq  \\ \leq \sum_{|G|,|g_m|} 4^{\sum |g_m|}2^c\delta^{\sum_m|g_m|}\frac{|g_1|(|g_1|+2|g_2|)(|g_1|+2|g_2| + 2|g_3|)\dots (2\sum_m |g_m|-|g_1|)}{C(|G|)}
\end{align}
Or if we separate sums over $|G|,\sum_m |g_m|=X,|g_m|=x_m$:
\begin{equation}
\|J_i\|\leq   \sum_{|G|} \sum_{X\geq |G|} 4^{X}2^c\delta^{X} \sum_{x_m| \sum_m x_m =X}\frac{x_1(x_1+2x_2)(x_1+2x_2 + 2x_3)\dots (2X-x_1)}{C(|G|)}
\end{equation}

We see that the smallest power of $\delta$ is given by $\sum_m |g_m| = |G|$, in which case the $C(|G|)$ exactly cancels its numerator:

\begin{equation}
\|J_i\|\leq   \sum_{|G|}  4^{|G|}2^c\delta^{|G|} F(|G|,\delta)
\end{equation}
where
\begin{align}
F(|G|,\delta)=1 +\\ +\sum_{X> |G|} 4^{X-|G|}\delta^{X-|G|} \sum_{x_m| \sum_m x_m =X}\frac{x_1(x_1+2x_2)(x_1+2x_2 + 2x_3)\dots (2X-x_1)}{C(|G|)}
\end{align}
we want to show that  $f(|G|,\delta)$ can be bounded by a constant for sufficiently small $\delta$. It is a series in terms of $n=X- |G|$. $F = \sum_{n=0}^{\infty}F_n$ and:
\begin{equation}
F_n =(4\delta)^n \sum_{x_m| \sum_m x_m =|G|+n}\frac{x_1(x_1+2x_2)(x_1+2x_2 + 2x_3)\dots (2(|G|+n)-x_1)}{C(|G|)}
\end{equation}
There are always $|G|$ terms $x_m$. One can already see that
\begin{equation}
\sum_{x_m| \sum_m x_m =|G|+n}\frac{x_1(x_1+2x_2)(x_1+2x_2 + 2x_3)\dots (2(|G|+n)-x_1)}{C(|G|)} \leq (|G| +n)^{2|G|}
\end{equation}
so we only get a polynomial prefactor in front of $\delta^n$ and the series converges. What we want to show is that  it converges to something that goes to $\infty$ only exponentially in $|G|$, so we can suppress it by $(4\delta)^|G|$. Clearly the naive bound $(|G| +n)^{2|G|}$ where dropped $C(|G|)$ altogether gives a factorial-like scaling even for the first term in the sum that we know is $1$. We know that number of terms in the sum for each $n$ is in fact given by $C_{|G|+n}^{|G|}< 2^{|G|+n}$, so what's left is to bound a single term in the sum:
\begin{align}
\frac{x_1(x_1+2x_2)(x_1+2x_2 + 2x_3)\dots (2(|G|+n)-x_1)}{C(|G|)} \leq \\ \leq \left(\frac{(2(|G|+n)-1}{2|G|-1}\right)^{|G|} = \left(1 +\frac{2n}{2|G|-1}\right)^{|G|}
\end{align}
these terms summed with $4\delta^n$ can only lead to logarithmic corrections $(1/(-\text{ln}\delta))^{|G|}$ to the $F(|G|,\delta$. Putting it all together:
\begin{equation}
F_n \leq(4\delta)^n 2^{|G|+n}\left(1 +\frac{2n}{2|G|-1}\right)^{|G|}
\end{equation}
There is no factorial in $|G|$, so:
\begin{equation}
F (|G|,\delta) = \sum_{n=0}^{\infty}F_n \leq c'= O(1)
\end{equation}
dropping the logarithmic corrections, and
\begin{equation}
\|J_i\|\leq \sum_{|G|}c' 2^c (4\delta)^{|G|} \leq c' 2^{c+1} (4\delta)
\end{equation}
Each term in the sum is bounded by our choice of $C(|G|)$. We can also consider long paths of $|G|>L_2$ like in the toy example 2 and get $ c' 2^{c+1} (4\delta)^{L_2}$.

\paragraph{Actual construction}

The structure of $J_i$ that we described is present in the generators $A_i$ of steps 1 and 2 of Imbrie's construction, with several differences. One difference is that in the original bound on $A_{\sigma\sigma'}(g)$ the $\gamma/\epsilon$ stands in place of $\delta$ that we used here. Another difference with $J_i$ is that the paths for $A$ are bounded in length: from below which leads to a higher power of $\gamma/\epsilon$ and from above by a growing scale $L_{k+1}$, which gives truly local generators, but does not in any way affect the derivation above. Also, the paths in $g^{(k)}$ do not have to touch on a chain - instead, a collar is allowed, which is then lifted to fill in the gap with new graphs as described on pp. 45-46 of paper\cite{Imbrie}. We discuss how the collar appears in the Subsection \ref{supoDance}.

Finally, in Imbrie's paper $C(x)\ne (2x-1)!!\sim (x!)^{1/2} $, instead it's an arbitrary power $C(x) \sim (x!)^{2/9}$ (the graph factorial is defined on p.23 in the same way as our toy examples). The combinatorial factors are still controlled, but they require a more careful way to bound $\sum_g$ and  extra tricks like resumming the long paths as a subroutine, as described on p. 27 of paper\cite{Imbrie}. We note that in the earlier arXiv version of the paper, the $n!^{1/2}$ was used as a combinatorial factor.

Here we repeat the argument on p.27. The construction results in $H
_{pert}=\sum_g J(g)$. Then $\langle \sigma | H_{pert} |\sigma'\rangle  =\sum_g J_{\sigma\sigma'}(g)$ where the terms in the sum have a bound proven in\cite{Imbrie} for them. We'd like to prove that the sum is bounded as well. For that we do a transformation $\sum_g |J_{\sigma\sigma'}(g)| =\sum_g |J_{\sigma\sigma'}(g)| c(g) c^{-1}(g) \leq  $ max$_gJ_{\sigma\sigma'}(g)| c(g) \sum_g c^{-1}(g)$ where $c(g)$ is some arbitrary function of the graph $g$. If $c(g)$ is chosen in such a way that $\sum c(g)^{-1}\leq 1$, then $\sum_g |J_{\sigma\sigma'}(g)| \leq $ max$_gJ_{\sigma\sigma'}(g)| c(g)$. Imbrie chooses $c(g)$ in such a way that max$_gJ_{\sigma\sigma'}(g)| c(g) \leq (a\chi) ^{\textrm{min}|g|}$ with some constant $a$, and proceeds to show that  $\sum c(g)^{-1}\leq 1$ for that choice. Since we consider $J_{\sigma\sigma'}(g)$, it is nonzero only for $g$'s that connect $\sigma$ and $\sigma'$, the min$|g|$ is at least the length of the interval containing the difference $\sigma-\sigma'$. Similar process can be applied to $A$, where  min$|g|$ is the minimal  remaining order ceil$L_{k-1}$.

\subsection{Bound on support}
\label{supoDance}

Imbrie's diagonalization proceeds in steps labeled by $k=1,2,3\dots$. At each step a unitary rotation $O_re^{A^{(k)}}$ is applied to the Hamiltonian, where $O_r$ is an operator acting spatially only near resonances, and as identity everywhere else. The generator of non-resonant part of the circuit is given as a sum of local terms $A^{(k)} = \sum_g A(g)$, where $g$ is the data in the format of the generalized path, containing, in particular, the locations of spin flips of $A(g)$. The number of spin flips is upper bounded by the order of the generalized path $|g|$ ~--- the power of the magnetic field involved in this term (some flips can act on the same location and "cancel" each other). At step $k$, that order is $L_{k-1}\leq|g| < L_{k}$, where $L_k= (15/8)^k$ (see, for example, p.24 of the new arXiv version). One may group terms into $A^{(k)}_i = \sum_{g:i\in\textrm{supp}A(g)}A(g)$ where the first site on which $A(g)$ acts nontrivially is $i$. We demonstated in Subsection \ref{normToy} how  $\|A^{(k)}_i\|$ constructed this way can be bounded knowing the bound on $A(g)$. Here we'll discuss the support of $A(g)$ and $A^{(k)}_i$. 

The construction has a separate term $A(g)^{(k)}$ for each matrix element $A(g)^{(k)}\sim |\sigma\rangle \langle \sigma'|$, given by $A(g)^{(k)} = \frac{J(g)^{(k-1)}}{E_\sigma^{(k-1)} -E_{\sigma'}^{(k-1)}}$. Here $J^{(k-1)}$ are the terms appearing after $k-1$ previous rotations of $H$. From the point of view of support, $J(g)^{(k-1)}$ combines support of the lower order paths that are part of $g$ data, while $E_{\sigma}^{(k-1)}$ is calculated using the diagonal part of the Hamiltonian truncated to terms $J(g)_{diag}^{(k-1)}$ of order less than $L_{k-1}$ (see p. 48 of paper\cite{Imbrie}. Such terms are diagonal, thus they have to wrap onto itself at least once to cancel all the flips. Using that, Imbrie proves a strong bound on extra support on each side $c_k=\frac12$ supp$E^{(k-1)}_\sigma - E^{(k-1)}_{\sigma'}\leq \frac{15}{14}L_{k-1}$. So the term we get in $A(g)$ has the same positions of spin flips as $J(g)$, and a collar around those spin flips of size $\frac{15}{14}L_{k-1}$.

Let's now discuss the support of $J^{(k-1)}(g)$ included in $A^{(k)}(g)$. The orders of terms $A^{(k)}(g)$ are in $[L_k-1,L_k)$.
The order $|g|$ limits the size the interval containing its spin flips. Thanks to bridging of the gap described on on pp. 45-46 of paper\cite{Imbrie}, order directly translates into maximal size of the interval where flips are happening (it is exactly $|g|<L_k$ for terms in $A^{(k)}$, and $|g|/2\leq L_{k-1}/2$ for diagonal terms in $J(g)_{diag}^{(k-1)}$ appearing in the denominator for $A^{(k)}$). So the total support is either $|g|$ or $|g|/2$ plus two collars $c_k$. We illustrate how the collars appear by explicitly listing the properties of terms for the first few steps. 
\begin{center}
\begin{tabular}{ c c c }
 $L_0=1$ & perturbation order $[1,15/8]$ is removed at step 1 by $A^1 ~ \Rightarrow$ & $H^1 + J^1$ \\ 
 ~ & ~ & $J^1$ starts with order 2 \\  
 $L_1=15/8$ & perturbation order $[2,3]$ is removed at step 2 by $A^2 ~ \Rightarrow$ & $H^2 + J^2$ \\
  ~ & $H^2$ includes order 2, support 3 terms & $J^2$ starts with order 2 \\ 
   $L_2\approx 3.5$ & perturbation order $[4,6]$ is removed at step 3 by $A^3_{c=2} ~ \Rightarrow$ & $H^3 + J^3$ \\
  ~ & $H^3$ includes order 6, support 7 terms & $J^3$ starts with order 7 \\ 
     $L_3\approx 6.6$ & perturbation order $[7,12]$ is removed at step 4 by $A^4_{c=5} ~ \Rightarrow$ & $H^4 + J^4$ \\
  ~ & $H^4$ includes order 12, support 16 terms & $J^4$ starts with order 13 \\ 
$L_4\approx 12.4$ & $c = 6+3 + 1+1 = 11$ & ~   
\end{tabular}
\end{center}
We see that even though in $H^3$ the support 7 terms are present, somehow the collar for $A^4$ is chosen to be 5. In other words, the path in the collar at every step should be touching, otherwise the collars blow up. We assume that terms with bigger collar are not included in each $A(g)^{(k)}$. Let us see what is the order of the extra terms that would appear if we tweak the procedure this way ~--- if it is the same order as $J^{(k)}$, then we can just include them with other paths.

For illustration, let's use denominators within a shorter collar in $A^4$, and check that  $e^A(H^3 + J^3) e^{-A} = H^4 + J^4$ with the same condition on $J^4$ to include only $|g|>13$. Recall that we want to remove terms in $J^3$ of orders $7$ to $12$. Now in principle  $H^3$ contains order $6$, support $7$ terms, that look like 2-site collar around the graph of step $3$. The collar itself is a step $2$ graph of support $3$, order $2$, with collar of size $1$ around it. It doesn't go further down. Now after conjugation, $[A,H^3]$ should cancel terms in $J^3$ in such a way that only terms of order $>12$ remain. But instead of cancellation, it gives a $(1- \Delta E_{full}/\Delta E_{short})J^3_{7..12}$ remainder $\Delta E_{full} -\Delta E_{short}$ contains terms in $H^3$ whose center graph does not touch the spins flipped by $A$. In this case, it is only the center size $3$, collar $2$ (so support 7) term of order 6. $7+6 =13$ which is enough to send it to $J^4$, and it's not even enough to make it connected - one needs to lift an order $2$ term to make it continuous. Unfortunately this activity does not really fit into the structure that Imbrie proves bounds for, but it can be done at least in principle. Since Imbrie makes no mention of this, we leave the exact way it's done as an open question and proceed. 

Here's the table at $k$'th step:
\begin{center}
\begin{tabular}{ c c c }
     $L_{k-1} = (15/8)^{k-1}$ & pert. order $[\textrm{ceil}(L_{k-1}),\textrm{floor}(L_k)]$ removed at step $k$ by $A^k_{c_k} ~ \Rightarrow$ & $H^k + J^k$ \\
  ~ & $H^k$ includes order floor$L_k$, support floor$(L_k/2)+ 2c_k$ terms & $J^4$ starts with order floor$L_k+1$ \\ 
$L_k\approx (15/8)^k$ & $c_{k+1} = \textrm{floor}(L_k/2)+ c_k$ & ~   
\end{tabular}
\end{center}
In other words $c_{k+1} = 1 +\sum_{k'=2}^k\textrm{floor}(L_{k'}/2) \leq \sum_{k'=1}^kL_{k'}/2$ which is the expression used by Imbrie. Let's evaluate it: 
\begin{equation}
c_{k+1} = \frac{1}{2}\sum_{k'=1}^k (15/8)^{k'} = \frac{1}{2}(15/8)^{k}\sum_{k'=0}^{k-1} (8/15)^{k'} \leq \frac{1}{2}L_{k}\sum_{k'=0}^{\infty} (8/15)^{k'}  = \frac{1}{2}L_{k} \frac{1}{1 - 8/15} = \frac{15}{14} L_k
\end{equation}
This is exactly as Eq 4.50 on p.46. So the support of a term in $A^k$ is bounded by floor$L_k + 2c_k \leq L_k + \frac{15}{7} L_{k-1} = \frac{15}{7}L_k = \frac87 L_{k+1}$. And the order is bounded by $L_{k-1}$ from below as seen in Subsection \ref{normToy} (even $k=1$ works out since ceil$L_0=1$). The final piece we need for the postulates in the main text is the support of resonances: after perturbative rotations $A^k$ of step $k$, there come nonperturbative rotations that may flip spins within an interval $3L_k$ (p.15) and have a collar $15/14L_j = 15/14L_{k-1}$ (p.25) on each side. So total support is $3L_k + \frac87 L_k = 4\frac17 L_k\leq 4.2 L_k$.

\end{document}